\documentclass[a4]{scrartcl}
\usepackage{cite}
\usepackage{amssymb}
\usepackage{sectsty}
\usepackage[utf8]{inputenc}
\usepackage[english]{babel}
\usepackage{graphicx}
\usepackage{amsmath}
\usepackage{amsthm}
\usepackage{csquotes}
\usepackage{hyperref}
\usepackage{cleveref}
\usepackage[draft]{todonotes}
\usepackage[ruled,vlined]{algorithm2e}
\usepackage{mathtools}
\usepackage{amsfonts}
\usepackage{enumerate}
\usepackage{comment}
\usepackage{thmtools}
\usepackage{thm-restate}
\usepackage{multicol}
\usepackage[affil-sl]{authblk}

\newtheorem{theorem}{Theorem}[section]
\newtheorem{corollary}[theorem]{Corollary}
\newtheorem{lemma}[theorem]{Lemma}
\newtheorem{fact}[theorem]{Fact}
\newtheorem{observation}[theorem]{Observation}

\newcommand{\mbbR}{\mathbb{R}}
\newcommand{\mcO}{\mathcal{O}}
\newcommand{\OPT}{\mathrm{OPT}}
\newcommand{\OPTS}{\ensuremath{\operatorname{OPT}}}
\newcommand{\CORE}{\ensuremath{\operatorname{core}}}
\newcommand{\core}{\ensuremath{\operatorname{core}}}

\newcommand{\CMove}{\ensuremath{\operatorname{cost}_{\operatorname{move}}}}
\newcommand{\CMerge}{\ensuremath{\operatorname{cost}_{\operatorname{merge}}}}
\newcommand{\CMono}{\ensuremath{\operatorname{cost}_{\operatorname{mono}}}}

\newcommand{\dbar}{\bar{\delta}}
\newcommand{\vio}{D}

\DeclarePairedDelimiter\norm{\lVert}{\rVert}
\def\smin{\mathrm{smin}}
\def\fmin{\mathrm{smin}}

\def\deltabar{\bar{\delta}}
\def\cost{\operatorname{cost}}
\def\costhit{\operatorname{cost}_{\operatorname{hit}}}
\def\costmove{\operatorname{cost}_{\operatorname{move}}}
\def\costmerge{\operatorname{cost}_{\operatorname{merge}}}
\def\costmono{\operatorname{cost}_{\operatorname{mono}}}
\def\costmig{\operatorname{cost}_{\operatorname{mig}}}

\long\def\harry #1{{\color{red}\bfseries\mathversion{bold}[#1 -- Harry]\mathversion{normal}}}
\long\def\ruslan #1{{\color{blue}\bfseries\mathversion{bold}[#1 -- Ruslan]\mathversion{normal}}}
\long\def\stefan #1{{\color{orange}\bfseries\mathversion{bold}[#1 -- Stefan]\mathversion{normal}}}

\long\def\harry #1{}
\long\def\ruslan #1{}
\long\def\stefan #1{}

\begin{document}

\title{Polylog-Competitive Algorithms for Dynamic Balanced Graph Partitioning for Ring Demands}


\author[1]{Harald R\"acke}
\author[2]{Stefan Schmid}
\author[1]{Ruslan Zabrodin}
\affil[1]{Technical University of Munich}
\affil[2]{Technical University of Berlin}




\date{}
\maketitle
\begin{abstract}
The performance of many large-scale and data-intensive 
distributed systems critically depends 
on the capacity of the interconnecting network. This paper is motivated by the vision of 
self-adjusting infrastructures whose resources can be adjusted according to the workload they currently serve,
in a demand-aware manner. 
Such dynamic adjustments can be exploited to improve network utilization and hence performance, 
by dynamically moving frequently interacting communication partners closer, e.g., collocating them in the same server 
or datacenter rack. 

In particular, we revisit the online balanced graph partitioning problem
which captures the fundamental tradeoff
between the benefits and costs of dynamically collocating communication partners.
The demand is modelled as a sequence $\sigma$ (revealed in an online manner) of communication requests between 
$n$ processes, each of which is running on one of the $\ell$ servers. Each server has
capacity $k=n/\ell$, hence, the processes have to be scheduled in a balanced manner across the servers.
A request incurs cost $1$, if the requested processes are located on different servers, otherwise
the cost is 0. 
A process can be migrated to a different server at cost $1$.

This paper presents the first online algorithm for online balanced graph partitioning achieving a
polylogarithmic competitive ratio for the fundamental case of ring communication patterns.
Specifically, our main contribution is a $O(\log^3 n)$-competitive randomized online algorithm for this problem.
We further present a randomized online algorithm which is $O(\log^2 n)$-competitive when compared to 
a static optimal solution. Our two results rely on different algorithms and techniques and hence are
of independent interest.
\end{abstract}

\section{Introduction}
\label{sec:intro}
Data-centric applications, including distributed machine learning, batch processing, scale-out databases, or streaming, produce a significant amount of communication traffic and their performance often critically depends on the underlying datacenter network \cite{mogul2012we}.
In particular, large flows (also known as elephant flows) may require significant network resources if communicated across multiple hops, resulting in a high ``bandwidth tax'' and consuming valuable resources which would otherwise be available for additional flows~\cite{mellette2017rotornet,griner2021cerberus}.

An intriguing approach to reduce communication overheads and make a more efficient use of the available bandwidth capacity, is to leverage the resource allocation flexibilities available in modern distributed systems (e.g., using virtualization), and render the infrastructures self-adjusting: by collocating two processes which currently exchange much data on the same server or datacenter rack, in a demand-aware manner, communication can be kept local and resources saved.    
When and how to collocate processes however is an algorithmically challenging problem, as it introduces a tradeoff: as migrating a process to different server comes with overheads, it should not be performed too frequently and only when the migration cost can be amortized by an improved communication later. Devising good migration strategies is particularly difficult in the realm of online algorithms and competitive analysis, where the demand is not known ahead of time. 

The fundamental algorithmic problem underlying such self-adjusting infrastructures is known as dynamic balanced graph (re-)partitioning and has recently been studied intensively \cite{soda21,apocs21repartition,sidma19,infocom21repartitioning,disc16,computing18,sigmetrics19learn,netys17learn}, see also the recent SIGACT News article on the problem \cite{sigact}. 
In its basic form, the demand is modelled as a sequence $\sigma$ (revealed in an online manner) of communication requests between $n$ processes, each of which is running (i.e., scheduled) on one of the $\ell$ servers. Each server has
capacity $k=n/\ell$, hence, the processes have to be scheduled in a balanced manner across the servers. A request incurs cost $1$, if both requested processes are located on different servers, otherwise the cost is 0.  A process can be migrated to a different server at cost $1$.
The goal is to design online algorithms which do not know $\sigma$ ahead of time, yet, they are competitive against an optimal offline algorithm with complete knowledge of the demand. 

Unfortunately, deterministic online algorithms cannot achieve a low competitive ratio: Avin et al.~\cite{disc16} (DISC 2016) presented a lower bound of $\Omega(k)$ for any deterministic algorithm for this problem, even
if the communication requests are sampled from a ring graph, and even in a resource augmentation model. Accordingly, most related work revolves around deterministic algorithms with polynomial competitive ratios 
\cite{apocs21repartition,sidma19,infocom21repartitioning,disc16,netys17learn}.
Hardly anything is known about the competitive ratio achievable by randomized online algorithms, except for a ``learning variant'' introduced by Henzinger et al.~in \cite{sigmetrics19learn} (SIGMETRICS 2019) and later studied by Henzinger et al.~in \cite{soda21} (SODA 2021): in this learning model, it is guaranteed that the communication requests can be perfectly partitioned, that is, the requests in $\sigma$ are drawn from a graph whose connected components can be assigned to servers such that no connected component needs to be distributed across multiple servers. For this learning variant, the authors presented a polynomial-time randomized algorithm achieving a polylogarithmic competitive
ratio of $O(\log \ell + \log k)$ which is asymptotically optimal. 
Unfortunately, however, these results on the learning variant are not applicable to communication patterns which do not perfectly fit into the servers, but which continuously require inter-server communication and/or migrations. 

Motivated by this gap, we in this paper study the design of randomized online algorithms for the balanced graph partitioning problem. 
We consider two models: 
\begin{itemize}
 \item \emph{Static:} The performance of the online algorithm is compared to an optimal static solution (for the given demand). This static model can be seen as a natural generalization of the learning variant discussed in prior work: while in prior work, it is assumed that a static solution exists which does not accrue any communication cost, in our static model, we do not make such an assumption on the communication pattern.
 \item \emph{Dynamic:} The performance of the online algorithm is compared to an optimal dynamic solution, which may also perform migrations over time.
\end{itemize}

As a first step, we consider a most fundamental setting where the demand $\sigma$ is chosen from a ring communication pattern. This is not only interesting because a ring pattern cannot be solved by existing algorithms designed for the learning variant (the ring is a large connected component which does not fit into a server) and because of the high lower bound for deterministic algorithms mentioned above, but also because of its practical relevance: machine learning workloads often 
 exhibit ring like traffic patterns~\cite{horovod, all-reduce, tree-reduce,griner2021cerberus}.

\subsection{Our Contributions}

We present the first online algorithm for online balanced graph partitioning achieving a
polylogarithmic competitive ratio for the fundamental case of ring communication patterns.
We first present a randomized online algorithm which is $O(\log^2 n)$-competitive when compared to 
an optimal static solution; this ratio is strict, i.e., without any additional additive terms.
Our second and main contribution is a $O(\log^3 n)$-competitive randomized online algorithm for this problem
when compared to an optimal dynamic algorithm.
Our two results rely on different algorithms and techniques and hence are
of independent interest.

\subsection{Related Work}
\label{onl:sec:related}

The dynamic balanced graph partitioning problem was introduced by Avin et
al. \cite{disc16,sidma19}. Besides the lower bound mentioned above (which even holds for ring graphs), they also present a deterministic online
algorithm which achieves a competitive ratio of $O(k \log k)$. Their algorithm  however relies on expensive repartitioning
operations and has a super-polynomial runtime. Forner et al.~\cite{apocs21repartition} later showed that a  competitive ratio of $O(k \log k)$
can also be achieved with a polynomial-time online algorithm which monitors the
connectivity of communication requests over time, rather than the density.
Pacut et al.~\cite{infocom21repartitioning} further contributed an
$O(\ell)$-competitive online algorithm for a scenario without resource
augmentation and the case where $k = 3$.
The dynamic graph partitioning problem has also been studied from an offline perspective by Räcke et al.~who presented a polynomial-time $O(\log n)$-approximation algorithm \cite{racke2022approximate}, using LP relaxation and Bartal’s clustering algorithm to round~it.

Deterministic online algorithms for the ring communication pattern have also been studied already, in a model where the adversary needs to
generate the communication sequence from a random distribution in
an \emph{i.i.d.} manner~\cite{obr-ring,netys17learn}: in this scenario, it has been shown that even deterministic algorithms can achieve a polylogarithmic competitive ratio. The problem is however very different from ours and the corresponding algorithms and techniques are not applicable in our setting. 

So far, to the best of our knowledge, randomized online algorithms have only been studied in the learning variant introduced by Henzinger et al.~\cite{sigmetrics19learn,soda21}. In their first paper on the learning variant, Henzinger et al.~\cite{sigmetrics19learn} still only studied deterministic algorithms and presented a deterministic exponential-time algorithm with competitive ratio $O(\ell \log \ell \log k)$ as well as a lower bound of
$\Omega(\log k)$ on the competitive ratio of any deterministic online algorithm.
While their derived bounds are tight for $\ell=O(1)$ servers, 
there remains a gap of factor $O(\ell \log \ell)$ between upper and lower bound
for the scenario of $\ell=\omega(1)$.
In~\cite{soda21}, Henzinger et al.~present deterministic and randomized algorithms which achieve (almost)
tight bounds for the learning variant. In particular, a
polynomial-time randomized algorithm is described which achieves a polylogarithmic competitive
ratio of $O(\log \ell + \log k)$; it is proved that no randomized online
algorithm can achieve a lower competitive ratio. Their approach establishes and
exploits a connection to generalized online scheduling, in particular, the
works by Hochbaum and Shmoys~\cite{hochbaum87using} and Sanders et
al.~\cite{sanders09online}. 

More generally, our model is related to 
dynamic bin packing problems which allow for limited \emph{repacking}~\cite{FeldkordFGGKRW18}:
this model can be seen as a variant of our problem where pieces (resp.~items)
can both be dynamically inserted and deleted, and it is also possible to open new
servers (i.e., bins); the goal is to use only an (almost) minimal number of servers, and
to minimize the number of piece (resp.~item) moves.
However, the techniques of~\cite{FeldkordFGGKRW18} do not extend to our problem.

\subsection{Organization}

The remainder of this paper is organized as follows.
We introduce our model and given an overview of our results in Section~\ref{sec:model}.
The dynamic model is studied in Section~\ref{sec:dyn} and the static model in  Section~\ref{sec:static}. We conclude our paper in  Section~\ref{sec:conclusion}. 

\section{Model and Results}\label{sec:model}

We formally define the dynamic balanced graph partitioning problem for ring demands, as
follows. Let $\ell$ denote the number of servers, and $k$ the \emph{capacity}
of a server, i.e., the maximum number of processes that can be scheduled on a
single machine. We use $P=\{p_0, p_1, \dots, p_{n-1}\}$ with $n\le\ell k$ to
denote the set of processes. Naming of processes is done modulo $n$, i.e.,
$p_i$ with $i \ge n$ refers to process $p_{i\bmod n}$.

In each time step $t$ we receive a request $\sigma_t=\{p_j, p_{j+1}\}$ with
$p_j,p_{j+1} \in P$, which means that these two processes communicate. Notice,
that this restricts the communication pattern to a cycle.

Serving a communication request incurs cost of exactly $1$, if both requested processes
are located on different servers, otherwise~0. We call this the
\emph{communication cost} of the request.
After the communication an online algorithm may (additionally) decide to
perform an arbitrary number of migrations. Each migration of a process to
another server induces a cost of $1$, and contributes to the
\emph{migration cost} of the request.

In the end after performing all migrations each server should obey its capacity
constraint, i.e., it should have at most $k$ processes assigned to it. The goal
is to find an online scheduling of processes to servers for each time step that
minimizes the sum of migration and communication cost and obeys the capacity
constraints.

We will compare our online algorithms to optimum offline algorithms that on the
one hand are more powerful as they know the whole request sequence in advance,
but on the other hand are more restricted in the migrations that they are
allowed to perform.
Firstly, we use resource
augmentation. This means  the offline algorithms have to strictly obey the
capacity constraint, i.e., they can schedule at most $k$ processes on any
server, while the online algorithm may schedule $\alpha k$ processes on any
server, for some factor $\alpha > 1$. Then we say that the online algorithm
uses resource augmentation $\alpha$.

In this model we obtain the following result.
\begin{restatable}{theorem}{thmone} 
There is a randomized algorithm that solves the dynamic balanced graph partitioning
problem for ring demands with expected cost $O(\frac{1}{\epsilon}\log^3 k)\OPT + c$ and resource
augmentation $2+\epsilon$, where $\OPT$ is the cost of an optimal dynamic algorithm and $c$ is a constant
not dependent on the request sequence.
\end{restatable}
One disadvantage of this result is the additive constant $c$, which means that
the algorithm is not strictly competitive. Note that there may exist very long
request sequences that have a very low optimum cost. Then the above theorem
would not give good guarantees. 

In a second model we compare the performance of an online algorithm (that again
uses resource augmentation) to that of
an optimal \emph{static} algorithm. Such an algorithm is only allowed to
perform migrations in the beginning before the first request arrives. In this
model our goal is to be \emph{strictly competitive}.
\harry{Why is it trivial otherwise?} We show the following
theorem.\harry{should we add more justification for the model here?}

\begin{restatable}{theorem}{thmtwo}
There is a randomized algorithm that solves the dynamic balanced graph partitioning
problem for ring demands with expected cost $O(\frac{1}{\epsilon^2}\log^2 k)\OPTS$ and resource
augmentation $3+\epsilon$, where
$\OPTS$ is the cost of an optimal static algorithm.
\end{restatable}

\section{The Dynamic Model}
\label{sec:dyn}
In this section we present an online polylog-competitive algorithm against a dynamic optimal algorithm. The main idea of our online algorithm is the reduction of the balanced graph partitioning problem for ring demands to the metrical task systems problem on a line. Metrical task system (MTS) is an online minimization problem that was first introduced in~\cite{borodin_mts} and has been extensively studied with various results for different types of the underlying metric~\cite{borodin_mts, bartal_mts_log6, fiat_unfair_mts_2003, bartal_mts_2006,bubeck_mts_log2}. There exist a tight $(2n-1)$ competitive deterministic algorithm~\cite{borodin_mts} and a tight $\mcO(\log^2 n)$ competitive randomized algorithm~\cite{bubeck_mts_log2,kserver_mts_lower_bound_log2}, if $n$ is the number of states in the system. 
\ruslan{say here something about the connection to MTS, cite?}

\subsubsection*{Notation}\ruslan{without subsection?}
For ease of exposition we model the problem of scheduling processes on servers
as a dynamic coloring problem, as follows. We identify each server $s$ with a
unique color $c_s$ and \emph{color} a process with the color of the server on
which it is currently scheduled.

We call the set of consecutive processes $S=\{p_s,\dots, p_{s+\ell-1}\}$ the \emph{segment} of length $\ell$ starting with $p_s$. For ease of notation we use $S=[s, s+\ell-1]$ to refer to this set.
We refer to a process pair $\{p_i,p_{i+1}\}$ as an \emph{edge} of the
cycle. To simplify the notation we denote such an edge with $(i,i+1)$.

For an algorithm $\mathrm{ALG}$ we use $\mathrm{ALG}(\sigma)$ to denote the cost
of the algorithm $\mathrm{ALG}$ on $\sigma$. However, usually $\sigma$ will be
clear from the context. Then we use $\mathrm{ALG}$ to denote both, the algorithm and
its cost.

\subsection{Algorithm}
Our strategy is to maintain a set of cut-edges which partition the cycle into
\emph{slices}. We ensure that we have at most $\ell$ slices, each of size at most
$(2+\epsilon)k$, such that we can map the slices directly to servers with
resource augmentation $(2+\epsilon)$. Each cut-edge is contained in an interval
$I$, i.e., its position is constrained to the interval. The problem of
choosing the cut-edge in some interval $I$ reduces to the metrical task systems
problem (MTS). The idea is to run a black box MTS algorithm for each interval
independently.
\harry{what is MTS?}
\stefan{
maybe mention the challenge here, so that it does not sound too simple? Like: such a specific MTS however first has to be developed?}

\def\ONL{\mathrm{ONL}}
Formally, let $k':=\left\lceil(1+\epsilon)k\right\rceil$ and $\ell'=\left\lceil\frac{n}{k'}\right\rceil$.
The algorithm $\ONL_R$ uses a \emph{shift parameter} $R\in\{0,\dots,k'-1\}$
that gives rise to intervals $I_1,\dots,I_{\ell'}$, where $\{I_i=[R+(i-1)k',R+ik']\}$.
Observe, that successive intervals share one vertex and intervals $I_{\ell'}$
and $I_1$ may also share edges. Each vertex is contained in at most two intervals.

\paragraph*{Reduction to MTS}\harry{define MTS}
A metrical task systems problem is defined as follows. We are given a metric
$(S, d)$ with $|S|=n$ states and a starting state $s_0 \in S$. Each time step
we receive a task $\sigma_i$ and a cost vector $T_i$, with $T_i(s)$ describing
the cost of processing the task $\sigma_i$ in state $s$. In response to
$\sigma_i$ an algorithm for the metrical task systems must choose a state
$s_i \in S$ and pay the cost $d(s_{i-1}, s_i) + T_i(s_i)$. The goal is to
minimize the overall incurred cost.

The online algorithm starts an MTS-instance for every interval. This instance
is responsible for choosing an edge inside the interval and is defined as
follows. Let $\sigma_I$ be the restriction of the request sequence to the
requests in $I$. For each interval $I$ we start an instance $M_I$ of an MTS
algorithm, where the states are the edges of the interval. On a request
$e \in \sigma_I$ we generate a cost vector $T$ with $T(e') = 1$, if $e'=e$ and
$T(e')=0$ otherwise. We forward the cost vector $T$ to the MTS instance $M_I$
and observe the new state $e'' \in I$. Subsequently, we move our cut-edge to $e''$.

\paragraph*{Server Mapping} The cut-edges inside the intervals induce a mapping
of processes to servers as follows.
Let $e_i=(a, a+1)$ and $e_{i+1}=(b,b+1)$ be the cut-edges chosen in intervals $I_i$
and $I_{i+1}$, respectively. We schedule the segment $[a+1, b]$ on the server
$s_i$. Observe, that a server receives exactly one slice. However, the slice
formed between $e_{\ell'}$ and $e_1$ could be empty because of the overlap
between intervals $I_{\ell'}$ and $I_1$.
\stefan{
bin nicht ganz sicher ob der ganze Algorithmus klar wird, dh. wie die Teile zusammenwirken. Könnte es Sinn machen, am Schluss nochmals einen Pseudocode zu geben?}
\subsection{Analysis}
We first argue that the schedule produced by the online algorithm roughly balances the
processes among the servers.
\begin{lemma}\label{lemma:dyn_load}
The load of each server is at most $2(1+\epsilon)k$.
\end{lemma}
\begin{proof}
A slice is the set of processes between the two cut-edges in two consecutive
intervals. Since the intervals contain at most $(1+\epsilon)k+1$ processes,
a slice can contain at most $2(1+\epsilon)k$ processes.
\end{proof}

The overall cost of the algorithm consist of two parts. The first part is
the \emph{communication cost}, which we denote by $\costhit$. The second part
is the \emph{migration cost}, which we denote by $\costmig$. 
\def\ALG{\mathrm{ALG}}
\subsubsection*{Interval Based Strategy}
For the analysis we introduce different types of (optimum) algorithms under
various restrictions and compare them to one another. The first concept is the
concept of an \emph{interval based strategy}. This is an algorithm that has to
choose a cut-edge in each interval and pays $1$ if the cut-edge is requested, and
$d$ if it moves the cut-edge by $d$ positions. For an interval $I$ and an
interval based strategy $\ALG$ we define
$\costhit^\ALG(I)$ as the total hit cost experienced by algorithm $\ALG$ on the
cut-edge maintained in interval $I$. Similarly, we define $\costmove^\ALG(I)$
as the cost for moving the cut-edge. Observe that our online algorithm $\ONL$
is an interval based strategy. \harry{ich weiss nicht ob das jetzt wirklich
  sauber ist...}
\begin{observation}
The communication cost and migration cost of the online algorithm
can be bounded by the interval cost. This means
\begin{itemize}
\item $\costhit \le \sum_I\costhit^\ONL(I)$
\item $\costmig \le \sum_I\costmove^\ONL(I)$
\end{itemize}
\end{observation}
\begin{proof}
If an edge is moved in an interval (by one position) the incident slices may change by at most
one process. This means at most one process has to be migrated to a different
server (note that due to the overlap it could happen that no slice changes when
an edge is moved). Similarly, whenever two processes located on different
servers want to communicate this communication takes place along a cut-edge.
This is a cut-edge in at least one interval (perhaps in two due to overlap) and
therefore this communication cost is counted in the interval cost.
\end{proof}
We define $\ONL_R:=\sum_I\costhit^\ONL(I)+\sum_I\costmove^\ONL(I)$ and use it as
a proxy for the cost of our online algorithm when using shift parameter $R$
(observe that the real cost could be lower due to the above observation).
Let $\OPT_R:=
\sum_I\costhit^{\OPT}(I)+
\sum_I\costmove^{\OPT}(I)$ denote the cost of
an optimal interval based strategy (that uses shift parameter $R$).

\begin{lemma}\label{lemma:cost_r}
For any $R$ we have $E[\ONL_R]\le\alpha(k)\cdot\OPT_R+c$, where $\alpha(k)$
is the competitive ratio of the underlying MTS algorithm on $k$ states, and $c$
a constant, that is independent of the request sequence.
\end{lemma}
\begin{proof}
Each subproblem on $I$ is essentially a metrical task systems problem on a line
metric. An interval $I$ consist of $k'$ edges that form the states. The
cut-edge $e$ is the current state. On a request $e \in I$, we generate the
cost vector $T$, where $T(e') = 1$ if $e'=e$ and 0 otherwise. Let $\sigma_{I}$
be the request sequence, constrained only to edges in $I$. Let
$\OPT_{\text{MTS}}(I)$ be the cost of an MTS algorithm (on the line $I$), that
serves $\sigma_{I}$ optimally. By definition this is equal to $\costhit^{\OPT_R}(I)+
\sum_I\costmove^{\OPT_R}(I)$.

Because we are using an $\alpha(k)$-competitive algorithm in each interval we get
$E[\ONL_R(I)]$ $\le\alpha(k)\cdot\OPT_{\text{MTS}}(I)+c'$ for some constant $c$ that does not
depend on the request sequence. Summing over all $I$ gives
$\ONL_R\le\alpha(k)\cdot\OPT_R+\ell c'$.
\end{proof}


\subsubsection*{Well Behaved Strategy}
Now we introduce a so-called \emph{well behaved} clustering strategy and
analyze its cost. We first show that a well behaved strategy can simulate the
real optimum with a small loss and then we show that an interval based strategy
with a random choice of $R$ can simulate the best well behaved strategy with a
small loss. Then the final result follows from Lemma~\ref{lemma:cost_r}.

We define a well behaved clustering strategy $W$ as a strategy that maintains a
set of cut edges $E_W=\{e_1, \dots, e_m\}$ which partition the cycle into
segments $S_1, \dots, S_m$. It can perform two operations:
\begin{itemize}
\item \emph{Move}: The cut edge $e=(i,i+1)$ is moved to $e'=(j,j+1)$ which
induces cost $|j-i|$. A \emph{merge} of two segments is simulated by a move
operation, where we move a cut edge $e$ to the position of another cut edge
$e'$ and remove $e$ from~$E_W$. 

\item \emph{Split}: The segment $S_i$ is split into several subsegments by
introducing new cut edges in $S_i$. The split operation induces no cost.
\ruslan{motivate here why the cost is zero?}
\end{itemize}
On a request $\sigma_i$, if $\sigma_i \in E_W$, then the algorithm has cost 1
(\emph{hitting cost}) or it moves the corresponding cut edge and pays the moved
distance (\emph{moving cost}). Furthermore, the size of a segment $S_i$ is
limited to $(1+\epsilon)k$. For technical reasons we assume
$\epsilon \le \frac{1}{4}$.

\def\W{\mathrm{W}}
The following crucial lemma shows that an optimal well behaved clustering
strategy is at most an $\mathcal{O}(\log k)$ factor away from $\OPT$.
\begin{lemma}\label{lemma:well_behaved}
	There exists a well behaved clustering strategy $\W$ with cost at most
        $\frac{4}{\epsilon}\log k\cdot\OPT + 2n\log k$.
        \end{lemma}

\begin{proof}
We develop a well behaved clustering algorithm $\W$ with the knowledge of
$\OPT$s choices. Let $c(p)$ be the color of the server where $\OPT$ places $p$
at the current time step. We denote by $E_O=\{e=(i,i+1) |~c(i) \neq c(i+1)\}$
the cut-edges of the optimal algorithm.

        \def\IH{\ensuremath{(\mathit{IH})}}
	\def\IM{\ensuremath{(\mathit{IM})}}
	\def\IS{\ensuremath{(\mathit{IS})}}

Let $\mathcal{S}_\W$ be the set of segments algorithm $\W$ produces. Segment
$S\in \mathcal{S}_\W$ is \emph{$\delta$-monochromatic}, if at least $\delta|S|$
processes in $S$ have the same color $c$. We fix
$\delta := \frac{1}{1+\epsilon}$. Each time $\OPT$ migrates a process, we mark
it. Marked processes will be a source of potential for our algorithm.
At some point we will remove marks from processes. We maintain the following invariants:
	\begin{itemize}
		\item[$\IH$] $E_\W \subseteq E_O$. 
		\item[$\IM$] All segments $S$ of algorithm $\W$ are $\delta$-monochromatic.
		\item[$\IS$] All processes of a segment $S$ with non-majority color are marked.
	\end{itemize}
	
	Due to invariant $\IH$ the hitting cost of $\W$ is at
        most the hitting cost of $\OPT$. Due to invariant $\IM$ the size of
        each segment $S$ is at most $(1+\epsilon)k$, since the number of
        processes with the majority color in $S$ is at least $\delta|S|$ and
        at most $k$. This gives $|S| \le (1+\epsilon)k$.

        At the start, the segments of $\W$ are essentially 1-monochromatic
        segments of the initial distribution; $E_W = E_O$ and there
        are no marked processes. On a new request $\sigma_t$ we mark all
        processes $\OPT$ migrated in time step $t$. Let $M_t$ be the number of
        marked processes after time step $t$ and let $k'=(1+\epsilon)k$.
	We introduce the potential function
        \begin{equation*}
	  \Phi_t=\frac{1+\epsilon}{\epsilon}\log(k') M_t + \sum_{S \in
          \mathcal{S}_W} |S|\log(\tfrac{k'}{|S|})\enspace.
        \end{equation*}

    Now, we analyze the amortized costs that we incur due to the movements of $\OPT$~and our adjustments. 
	
    \paragraph*{OPT Movement.} Let
    $o_t$ denote the number of newly marked processes in time step $t$. Then,
    the potential increases by
    $\Delta \Phi_{\mathrm{opt}} = \frac{1+\epsilon}{\epsilon}\log(k') o_t$. $\OPT$'s moving cost in this time step is at least $\sum_t o_t$.

    \bigskip
    After $\OPT$ performed its movements, we have to maintain the invariants.
    Let $e_j \in E_\W\setminus E_O$. The edge $e_j$ separates two segments, $L$
    and $R$, in for algorithm $\W$. Prior to request $\sigma_t$, the segments $L$ and $R$ were
    $\delta$-monochromatic due to invariant $\IM$. Let $c_L$ and $c_R$ be their
    majority colors, respectively. For each $e_j \in E_\W\setminus E_O$ we
    perform a merge, a move, or a cut-out operation.
	
    \paragraph*{Merge Operation.}If $c_L=c_R$, we merge the segments $L$ and $R$ (wlog.
    $|L| \le |R|$) and pay $c_{\textit{merge}}=|L|$. The new segment $S=L\cup R$ has the
    size $|S|=|L|+|R|$.
	Notice, that the invariant $\IS$ is maintained, we do not remove any marks. Let $|L|=l$ and $|R|=r$. Our change in the potential function is then:
        \begin{equation*}
        \begin{split}
	\Delta \Phi_{\textit{merge}} &= |S|\log(\tfrac{k'}{|S|}) - r\log(\tfrac{k'}{r}) - l\log(\tfrac{k'}{l})\\
	&=  l(\log(\tfrac{k'}{l + r}) - \log(\tfrac{k'}{l})) + r(\log(\tfrac{k'}{l + r}) - \log(\tfrac{k'}{r}))\\
	&= l\log(\tfrac{l}{l + r}) + r\log(\tfrac{r}{l + r})\\
	&= l\log(1-\tfrac{r}{l + r}) + r\log(1-\tfrac{l}{l + r})\\
	&\le l(-\tfrac{r}{l + r}) + r(-\tfrac{l}{l + r})\\
	&= -l(\tfrac{2r}{l + r})\\
	&\le -l
        \end{split}
\end{equation*}
The first inequation is due to the well known fact $\log(1+x) \le x$. The second inequality holds because $l \le r$.
	Thus, the amortized cost of the merge operation is $\tilde{c}_{\textit{merge}}=c_{\textit{merge}} + \Delta \Phi_{\textit{merge}} \le 0$.

	\bigskip
        Now suppose that $c_L\neq c_R$. Let $e_l=(l, l+1) \in E_O$ be the
        nearest cut edge of $\OPT$ "left" of $e_j$, and $e_r=(r, r+1) \in E_O$
        the nearest cut edge of $\OPT$ "right" of $e_j$, respectively. Let
        $F=[l+1, r]$. By construction $F$ contains only processes with the same
        color $c$. We differentiate whether one of the colors $c_L, c_R$ is
        equal to $c$ or not. We need the following fact for the analysis.

	\begin{fact}\label[]{fact:f_d}
		Let $f(d) = (s-d)\log(\frac{s}{s-d})$ with constants $s\ge 2$, $1\le d < s$.
		Then, $f(d) \le d$.
	\end{fact}
	\begin{proof} We have
			$f(d) = (s-d)\log(\frac{s}{s-d})
			= (s-d)\log(1+\frac{d}{s-d})
			\le (s-d)(\frac{d}{s-d})
			= d$,  where
		the inequality is due to the well known fact $\log(1+x) \le x$.
	\end{proof}
	
	\paragraph*{Move Operation.} If $c_L=c$, all processes in $F\cap R$ must be marked
        due to invariant $\IS$. Then, we move $e_j$ to $e_r$, a distance
        $d=|F\cap R|$ and remove the marks of $F\cap R$. If $c_R=c$, we act
        analogously. Our actual cost is $c_{\textit{move}}=d$. Let $|L|=l$ and $|R|=r$.
        Then, the change in potential is:

	\begin{equation*}\begin{split}
          \Delta \Phi_{\textit{move}} &= [(l+d)\log(\tfrac{k'}{l+d}) - l\log(\tfrac{k'}{l})]
                             + [(r-d)\log(\tfrac{k'}{r-d}) - r\log(\tfrac{k'}{r})] -d \tfrac{1+\epsilon}{\epsilon}\log k'\\
                             &=  l\log(\tfrac{l}{l+d}) + r\log(\tfrac{r}{r-d})
                             + d\log(\tfrac{k'}{l+d})-d\log(\tfrac{k'}{r-d})-d\tfrac{1+\epsilon}{\epsilon}\log k'\\
                             &\le r\log(\tfrac{r}{r-d}) - d\log(\tfrac{k'}{r-d})-d\log k'-d\log(l+d)\\
                             &\le (r-d)\log(\tfrac{r}{r-d})-d\log k'-d\log(l+d)\\
                             &\le d-d\log k'-d\log(l+d)
                             \le -d
\end{split}\end{equation*}
The first inequality holds because $\log(\frac{l}{l+d}) < 0$ and
$\frac{1+x}{x} \ge 2$, for $0<x \le 1$. The second last inequality holds due to
Fact~\ref*{fact:f_d}. The last inequality holds since $k'\ge 2$ and
$l+d\ge 2$. Hence, the amortized cost of the move operation is
$\tilde{c}_{\textit{move}}=c_{\textit{move}} + \Delta \Phi_{\textit{move}} \le 0$.
	
\paragraph*{Cut-out Operation.} Now, assume that $c \neq c_L$ and $c \neq c_R$.
Then, all $p\in F$ must be marked. Wlog.\ let $|j-l| \le |j-r|$. We move
$e_j$ to $e_l$ and make a split in $R$ by creating a new cut edge $e_r$. 
$F$ now becomes a new segment for algorithm $\W$. We pay  at most $|F|/2$ for the movement of $e_j$.
Afterwards, we remove the marks from the processes in $F$. Notice, that the
new segment $F$
is 1-monochromatic. Furthermore invariants $\IM$ and $\IS$
are also maintained. Let $|L|=l$, $|R|=r$, $|F|=d$, $|j-l|=d_l$, and
$|j-r| = d_r$. The change in potential~is:
	\begin{equation*}\begin{split}
          \Delta \Phi_{\textit{cut-out}} &= \big[(l-d_l)\log(\tfrac{k'}{l-d_l}) - l\log(\tfrac{k'}{l}) \big]\\
                                         &\quad+ \big[(r-d_r)\log(\tfrac{k'}{r-d_r}) - r\log(\tfrac{k'}{r}) \big] + d\log(\tfrac{k'}{d}) - d \tfrac{1+\epsilon}{\epsilon}\log k' \\
                                         &\le \big[l\log(\tfrac{l}{l-d_l}) - d_l \log(\tfrac{k'}{l-d_l})\big]
                                         + \big[r\log(\tfrac{r}{r-d_r}) - d_r\log(\tfrac{k'}{r-d_r})\big] - 5d\log k'\\
		&\le (l-d_l)\log(\tfrac{l}{l-d_l}) + (r-d_r) \log(\tfrac{r}{r-d_r}) - 5d\log k'\\
		&\le d_l + d_r -5d\log k'\\
		&=d - 5d\log k'\\
		&\le -d
	\end{split}\end{equation*}
	The first inequality is due to $\frac{1+x}{x} \ge 6$, for
        $0<x \le \frac{1}{4}$. For the second inequality we use
        Fact~\ref*{fact:f_d}. Hence, the amortized cost of the cut-out
        operation is
        $\tilde{c}_{\textit{cut-out}}=c_{\textit{cut-out}} + \Delta \Phi_{\textit{cut-out}} \le 0$.

	\bigskip
	After performing the above adjustments we ensured that invariants $\IH$
        and $\IS$ hold, i.e., $E_W \subseteq E_O$ and all non-majority color
        processes are marked. We have to take care of the $\IM$ invariant that
        all segments are $\delta$-monochromatic. For this we use the split-operation.

	\paragraph*{Split Operation.} If there exists a segment $S$ in $\W$
        that is not $\delta$-monochromatic anymore, we simply make a "full"
        split. We create new cut-edges $S\cap E_O$, such that segment $S$
        breaks up into smaller 1-monochromatic pieces. Since $S$ is
        not $\delta$-monochromatic anymore, there are at least ($1-\delta$)
        marked processes in $S$ due to invariant $\IS$. We remove marks
        from these processes. Let $T_1, \dots, T_m$ be the resulting
        subsegments of $S$. The change in potential is:

	\begin{equation*}\begin{split}
		\Delta \Phi_{\textit{split}} &\le \big[\textstyle\sum_{i} |T_i|\log(\tfrac{k'}{|T_i|})- |S|\log(\tfrac{k'}{|S|})\big] - (1-\delta)|S|\tfrac{1+\epsilon}{\epsilon}\log k'\\
		&\le \log(k')\textstyle\sum_{i} |T_i|-|S|\log(\tfrac{k'}{|S|}) - |S|\log k'\\
		&= |S|\log k'-|S|\log(\tfrac{k'}{|S|}) - |S|\log k'\\
		&= -|S|\log(\tfrac{k'}{|S|})\\
		&\le 0
                \end{split}
	\end{equation*}
	The last inequality holds, because $|S| \le (1+\epsilon)k=k'$.
	Thus, the amortized cost of the split operation is $\tilde{c}_{\textit{split}}= \Delta \Phi_{\textit{split}} \le 0$.

	\bigskip
	From the above analysis we conclude, that in each time step $t$ our
        amortized cost is at most
        $\tilde{c}_t \le \frac{1+\epsilon}{\epsilon}\log(k')o_t$. Then, our
        overall cost $\W$ is at most
	$\W \le \frac{1+\epsilon}{\epsilon}\log(k') \textstyle\sum_t o_t + \Phi_0\le {2(1+\epsilon)^2}/{\epsilon}\cdot\log(k)\cdot\OPT + 2n\log k$.
	As $2(1+\epsilon)^2 \le 4$ for $\epsilon \le \frac{1}{4}$ the lemma follows.
\end{proof}

\begin{lemma}\label{lemma:region_well_beh}
Let $\OPT_{\W}$ be an optimal well behaved clustering strategy.
When choosing the shift parameter $R$ uniformly at random from $\{0,k'-1\}$
the optimum interval based strategy has expected cost
$E_R[\OPT_R]\le6\OPT_{\W}$.
\end{lemma}
\begin{proof}
We add additional cost of $k'$ to a well behaved strategy $\OPT_{\W}$ each time a
cut edge of $\OPT_{\W}$ crosses the boundary of a interval. Let $\OPT'_{\W}$ denote
this adapted overall cost.

There are $l'+1$ interval borders, i.e., the probability for a given process $p$ to be
an interval border is at most $\frac{l'+1}{n} \le \frac{2}{k'}$. If a cut edge
of $\OPT_{\W}$ moves a distance of $d$, then each process on the way is an
interval border with probability at most ${2}/{k'}$. This gives, that the
expected movement costs are at most $d+d \frac{2}{k'}k'=3d$. Hence, $E[\OPT'_{\W}] \le 3\OPT_{\W}$.

We develop an interval based algorithm that tries to mimic $\OPT_{\W}$. Fix
some interval $I$. Since $k'=(1+\epsilon)k$, we know that each well behaved
strategy should have always at least one cut-edge inside $I$. We choose one
such cut edge $e$ and every time $e$ is moved, we simply try to follow its
movement. If the movement takes $e$ outside of $I$, we know that $\OPT'_{\W}$ pays
cost $k'$. Then, instead of following $e$ we choose another cut-edge inside $I$ and travel there at a cost of at most $k'$.

By this construction all cost (hitting or movement) of the interval based algorithm can be charged to
the movement or hitting cost of an edge in $\OPT'_\W$. However, note that
intervals may overlap. This means there may be two intervals in the interval
based algorithm that charge against the same edge of $\OPT'_\W$. 
Nevertheless we still get $E_R[\OPT_R]\le2\OPT'_\W\le 6\OPT_\W$, as desired.
\end{proof}



\thmone*
\begin{proof}
We have
$E_R[\ONL_R] \le E_R[\alpha(k)\OPT_R+c']$ (Lemma~\ref{lemma:cost_r}), $E_R[\OPT_R] \le 6\OPT_\W$ (Lemma~\ref{lemma:region_well_beh}) and $\OPT_\W     \le \tfrac{4}{\epsilon'}\log k\cdot\OPT + 2n\log k$ (Lemma~\ref{lemma:well_behaved}), 
where $\alpha(k)$ denotes the competitive ratio of an algorithm
for metrical task systems on a line with $\mcO(k)$ states. By
using the algorithm presented in \cite{bubeck_mts_log2} (which achieves
competitive ratio $\mcO(\log^2 k)$ on any metric space), we conclude that the
cost of the online algorithm fulfills
$\ONL\le\mcO(\frac{1}{\epsilon'}\log^3 k)\OPT + c$, where $c$ is a constant not
dependent on the request sequence. By
Lemma~\ref{lemma:dyn_load} we know that the load of each server is at most
$(2+2\epsilon')k$. Setting $\epsilon'=\epsilon/2$ gives the theorem.
\end{proof}

\section{The Static Model}
\label{sec:static}

In this section we present an polylog-competitive online algorithm against an optimal static algorithm.
We start with a solution for a simpler problem that will serve as a main
ingredient for our algorithm.


\subsubsection*{Notation}
For ease of exposition we model the problem of scheduling processes on servers
as a dynamic coloring problem, as follows. We identify each server $s$ with a
unique color $c_s$ and \emph{color} a process with the color of the server on
which it is currently scheduled.

We call the set of consecutive processes $S=\{p_s,\dots, p_{s+\ell-1}\}$ the \emph{segment} of length $\ell$ starting with $p_s$. For ease of notation we use $S=[s, s+\ell-1]$ to refer to this set.
For a parameter $\delta$ we call $S$ \emph{$\delta$-monochromatic} for a color
$c$ if strictly more than $\delta|S|$ processes $p\in S$ are \emph{initially} colored
with $c$. We call a segment \emph{monochromatic} if it is
$\delta$-monochromatic for a value $\delta\ge 1/2$. In this case we call $c$ the
majority color of segment $S$. Let $\dbar := \max\{\frac{2}{2+\epsilon}, \frac{14}{15}\}$, for a $\epsilon > 0$.

We refer to a process pair $\{p_\ell,p_{\ell+1}\}$ as an \emph{edge} of the
cycle. To simplify the notation we denote such an edge with $(\ell,\ell+1)$.

\subsection{Hitting Game on the Line} \label{sec:hit_line}
In this section we analyze a simplified problem, that will server as a building block for our algorithm. 
The simplified problem is defined as follows.

A line of $k+1$ nodes
$V=\{v_1, v_2, \dots, v_{k+1}\}$ and $k$ edges $E=\{e_1, e_2, \dots, e_{k}\}$ with $e_i = \{v_i, v_{i+1}\}$ is given. Our initial position is the central edge $e_{s}, s=\lceil \frac{k}{2} \rceil$. Each time step $t$, $1\le t \le N$, we receive a request $e\in E$. If our current position is $e$, we may stay there and pay cost of 1 (the \emph{hitting cost}). Alternatively, we could change our position and pay the traveled distance (the \emph{moving cost}), where the distance between edge $e_i$ and $e_j$ is $d(e_i, e_j)=|i-j|$.

In our algorithm we use techniques of Blum et al.~\cite{finely_paging_1999} to maintain a probability distribution over the edge set. 
However, in~\cite{finely_paging_1999} they consider a uniform metric, whereas we have a line, i.e.,
they could switch between any two states with a cost of 1. We need to choose our probability distribution carefully to correctly incorporate our moving cost.

We compare our algorithm's cost against an optimal static strategy that chooses one position $e_{p}$ at the beginning, pays the distance $|s-p|$ and stays at this position for all subsequent requests. Let $\OPTS$ be the cost of such an optimal static algorithm. 
Our first observation is that we have to use a randomized strategy in order to be better than $\Omega(k) \OPTS$:
\begin{restatable}{lemma}{hittone}\label{lemma:lowe_bound_line}
Any deterministic online algorithm has cost at least $\Omega(k)$ $\OPTS$.
\end{restatable}
\begin{proof}
Since the adversary knows the position of the deterministic algorithm $\operatorname{DET}$, it
just requests its position each time. Then, after time step $T \ge k^2$ the
algorithm $\operatorname{DET}$ has cost at least $T$. On the other hand, by an averaging
argument there must be an edge that was requested at most $\frac{T}{k}$ times.
Traveling there costs at most $k$. Then, $\frac{T}{T/k + k}\ge \frac{k}{2}$.
\end{proof}

Observe, that a reasonably competitive algorithm should not move too far away from the starting position right at the beginning, since the optimal static algorithm could stay nearby to the starting position and pay only constant hitting cost. On the other hand, it has to react fast enough to increasing hitting cost.

Let $x^{(t)}$ denote the \emph{request vector}, with $x^{(t)}_e$ describing the number of requests to edge $e$ up until time step $t$
and let $x^{(t)}_I$ be the restriction of the vector $x^{(t)}$ to the edges inside an interval $I=\{v_l, \dots, v_r\}$. 
Note, that $x^{(0)}=0$.

\subsubsection*{Interval Growing Algorithm}
The \emph{interval growing algorithm} maintains an interval $I$ around the starting edge and restricts the possible positions only to this interval. 
Note that an interval $I$ contains $|I|-1$ edges.
The algorithm proceeds in phases. 
We begin with interval $I_0 = \{v_s, v_{s+1}\}$, i.e., only the starting edge is contained in $I_0$. We say $I_0$ is the \emph{initial} interval.
Let $I=[\ell,r]$ be our current interval. 
We denote by $\min(I)$ the minimum $\min_{e\in I}(x^{(t)}_e)$ at the current time step $t$.
Whenever $\min(I)$ reaches
$(1-\dbar)|I|$ we double the size of the interval and set
$I':=[\ell-{|I|}/{2},r+{|I|}/{2}]$ (and then choose a new edge
inside $I'$). Thus, a new phase begins with $I'$ as our new interval. There is one exception to this growth rule: when the length of
the new interval would be larger than $k+1$ we define
$I':=[\ell-\lceil(k+1-|I|)/2\rceil,r+\lceil(k+1-|I|)/2\rceil]$, i.e., we only
give the new interval a length of $k+1$.

\def\sminp{\fmin'}
During each phase we maintain a probability distribution over the edge set of our current interval.  
Let $I$ be our current interval.
We select a random edge inside $I$ according
to the probability distribution $p^{(t)}=\nabla\smin_{|I|-1}(x^{(t)}_{I})$ 
(see Appendix~\ref{appendix}). 

Let $\sminp(x)$ denote $\fmin_d(x)$ for a d-dimensional vectors $x$. 
On a new request our probability distribution changes from $p^{(t-1)}$ to
$p^{(t)}$, subsequently we have to move according to the change in the
probability distribution in order to maintain the invariant. At the end of a
phase, we grow our interval, and subsequently choose a new edge inside the new
interval $I'$, which incurs cost at most $|I'|$.

We denote by $\costhit^{(t)}$ and $\costmove^{(t)}$ the incurred hitting and moving cost at time $t$, respectively.
Let $\ell^{(t)}$ be a vector such that $\ell^{(t)}_i = 1$ if $e_i$ is the requested edge at time $t$ and $\ell^{(t)}_j = 0$ for $j \neq i$. Notice, that $x^{(t-1)} + \ell^{(t)} = x^{(t)}$.
Then, the expected hitting cost at time $t$ is at most $E[\costhit^{(t)}] = (p^{(t-1)})^T \ell^{(t)}$.
Since the distance between two edges is at most $k$, the expected moving cost is at most $k$ times the Earthmover distance between distributions $p^{(t-1)}$ and $p^{(t)}$, which in our case is at most $k\norm{p^{(t)}-p^{(t-1)}}_1$.
The advantage of using the $\nabla\sminp$ function is that we ensure that our moving cost is comparable to the hitting cost. 

\begin{restatable}{lemma}{hitttwo}
\label{lemma:hit_game_opt_lower_bound}
Let $I$ be the current, non-initial interval of the interval growing algorithm. 
Then, $\OPTS \ge \max\{\frac{1}{2}\min(I), \frac{1-\dbar}{2} |I|\}$.
\end{restatable}
\begin{proof}
Let $I'$ denote the
interval $I$ before the most recent growth step.
At the beginning $\OPTS$ makes its move and either stays in $I'$ and suffers at least $\min(I')$ communication cost or moves out of $I'$ and pays $\frac{1}{2} |I'|$ moving cost. Since $|I|=2|I'|$ and $\min(I') \ge (1-\dbar)|I'|$, we conclude that $\OPTS \ge \frac{1-\dbar}{2} |I|$.

Assume $I$ has the maximum possible size $k+1$. Then, $\OPTS \ge \min(I)$ is a trivial lower bound, since $\OPTS$ has to choose one position in $I$.
Otherwise, by construction $\min(I) \le (1-\dbar) |I|$ and since $\OPTS \ge \frac{1-\dbar}{2} |I|$ we conclude that  $\OPTS \ge \frac{1}{2} \min(I)$.
\end{proof}

Let $I$ be the current interval of the interval growing algorithm. 
We denote by $\costhit(I)$ the overall hitting cost and by  $\costmove(I)$ the overall moving cost incurred by the algorithm.
Then, $\cost(I) := \costhit(I) + \costmove(I)$ denotes the overall cost of the interval growing algorithm. 




\begin{restatable}{lemma}{hittthree}
\label{lemma:grow_alg}
Let $I$ be the current interval of the interval growing algorithm. The algorithm incurs the following costs:
\begin{enumerate}[a)]
	\item $E[\costhit(I)] \le 2\min(I) + \mathcal{O}(\ln |I|)~|I|$.
	\item $E[\costmove(I)] \le 4\min(I) + \mathcal{O}(\ln |I|)~|I|$.
	\item $E[\cost(I)] = 0$, if $I$ is the initial interval.
\end{enumerate}
\end{restatable}
\begin{proof}
If $I$ is the initial interval, then there was no request to the initial edge yet, no costs were incurred.

Now, we fix some phase and derive the overall hitting and moving cost incurred in this phase.
Let $I_i$ denote the interval maintained in phase $i$.
Assume that phase $i$ start at time $t_1$ and ends at time $t_2$.
We denote by $P^{(i)}_{\textit{hit}}$ and $P^{(i)}_{\textit{move}}$ the incurred hitting and moving cost during phase $i$, respectively. 
We know that $E[\costhit^{(t)}] \le (\nabla\sminp(x^{(t-1)}_{I_i}))^T \ell^{(t)}$. Furthermore, applying Lemma~\ref{lemma:smin} we derive that 
\begin{align*}
E[\costmove^{(t)}] &\le |I_i|~\norm{\nabla\sminp(x^{(t-1)}_{I_i}+\ell^{(t)})}\\
&\le 2(\nabla\sminp(x^{(t-1)}_{I_i}))^T \ell^{(t)}\enspace.
\end{align*}
Again, using Lemma~\ref{lemma:smin} we obtain:
\begin{align*}
	\sum_{t=t_1}^{t_2} \nabla\sminp(x^{(t-1)}_{I_i})^T \ell^{(t)} 
	&\le 2\sum_{t=t_1}^{t_2} \big(\sminp(x^{(t-1)}_{I_i}+\ell^{(t)})-\sminp(x^{(t-1)}_{I_i})\big)\\
	&= 2\sum_{t=t_1}^{t_2} \big(\sminp(x^{(t)}_{I_i})-\sminp(x^{(t-1)}_{I_i})\big)\\
	&= 2\big(\sminp(x^{(t_2)}_{I_i}) - \sminp(x^{(t_1-1)}_{I_i})\big)\\
	&\le 2\big(\min(x^{(t_2)}_{I_i}) + |I_i|\ln |I_i|\big)\\
\end{align*}
Then, $P^{(i)}_{\textit{hit}} \le 2(\min(x^{(t_2)}_{I_i}) + |I_i|\ln |I_i|)$ and 
$P^{(i)}_{\textit{move}} \le 4(\min(x^{(t_2)}_{I_i}) + |I_i|\ln |I_i|)$.

Let $m$ be the current phase. 
We know that for $i < m$ the minimum $\min(I_i)$ at end of phase $i$ is exactly $(1-\dbar)|I_i|$. 
Then, $\min(I_i) + |I_i|\ln |I_i| \le (\ln |I_i| + 1) |I_i|$. 
The moving costs of the interval growing algorithm are essentially the costs for maintaining the probability distribution. Additionally, at the end of each phase $i$ we pay at most $|I_{i+1}|$ for choosing a new position in the interval $I_{i+1}$. Then:
\begin{align*}
\costmove(I) &=\sum_{i=1}^{m-1} (E[P^{(i)}_{\textit{move}}] + |I_{i+1}|) + E[P^{(m)}_{\textit{move}}]\\
&\le \left[ \sum_{i=1}^{m-1} 4(\ln |I_i| + 1) |I_i| + |I_{i+1}| \right] + E[P^{(m)}_{\textit{move}}]\\
&\le \left[ \mathcal{O}(\log |I|) \sum_{i=1}^{m-1} |I_i| \right]+E[ P^{(m)}_{\textit{move}}]\\
&\le \mathcal{O}(\log |I|)~|I| + 4 (\min(I) + |I|\ln |I|) \\
&\le 4 \min(I) + \mathcal{O}(\log |I|)~|I| \\
\end{align*}
Using the same steps we derive that the hitting cost $\costhit(I)$ of the interval growing algorithm is at most $2\min(I) + \mathcal{O}(\log |I|)~|I|$.
\end{proof}

\begin{corollary}
The expected cost of the interval growing algorithm $E[\cost(I)]$ is at most $\mcO(\frac{1}{1-\dbar}\log k) \OPTS$.
\end{corollary}

\subsection{Algorithm}

The algorithm consists of three procedures. The \emph{Slicing Procedure}, the
\emph{Clustering Procedure} and the \emph{Scheduling Procedure}.
\begin{enumerate}
\item
The \emph{Slicing Procedure} gets the
request sequence as input and maintains a set of \emph{cut edges}. These are
edges in the cycle for which the two endpoints \emph{may} be scheduled on
different servers (to be determined by the clustering and scheduling
procedure). The cut edges partition the cycle into \emph{slices} (segments that
start and end in a cut-edge). During the
algorithm the set of slices changes dynamically as cut edges are 
moved or deleted (a deletion of a cut edge causes a merge of the two incident segments).

\item
The input to the \emph{Clustering Procedure} is the sequence of slice change
operations as generated by the Slicing Procedure. The
Clustering Procedure maintains a grouping of slices into clusters such that
slices that almost exclusively contain processes with a particular (initial)
color are grouped together (these are $\frac{3}{4}$-monochromatic slices).

\item
Finally, the \emph{Scheduling Procedure} schedules the clusters on the servers.
It makes sure that clusters that contain monochromatic slices for a particular
server $s_c$ are scheduled on this server.
\end{enumerate}

\subsubsection*{The Slicing Procedure}
The basic idea of the Slicing Procedure is to maintain a set of cut-edges by
using the interval growing algorithm from Section~\ref{sec:hit_line}. Initially, the cut
edges are edges in the cycle for which both end-points are on different servers
w.r.t.\ the initial distribution of processes. We create an interval
$I=[i, i+1]$ around each initial cut edge $e=(i, i+1)$. We call $e$ the
\emph{center} of interval $I$. The idea is to run the interval growing
algorithm for each interval independently.

Clearly, this na\"ive approach does not work. One major challenge is that the
analysis for the Hitting Game (see Lemma~\ref{lemma:grow_alg}) compares the
cost of the online algorithm to the cost of an optimum algorithm inside the
interval. In order to salvage this analysis we have
to make sure that the intervals do not overlap too much. A second difficulty
is that the lower bound on $\OPT$ for Lemma~\ref{lemma:hit_game_opt_lower_bound} hinges on the fact that
$\OPT$ needs to choose at least one edge. However, if the initial distribution
contains many cut-edges then $\OPT$ does not need to choose a cut edge for
every interval. 

Formally, we transform the Hitting Game approach to the set of intervals as follows.
Let $x$ denote the \emph{request vector}, with $x_e$ describing the number of
requests to edge $e$ and let $x_I$ be the restriction of the vector $x$ to the
edges inside an interval $I$.
We select a random cut-edge inside each interval according
to the probability distribution $\nabla\sminp(x_{I})$ (where $\sminp(x)$ denotes $\fmin_d(x)$ for a d-dimensional vectors $x$; 
see Appendix~\ref{appendix}).
Whenever for an interval $[\ell,r]$ the minimum $\min_{e\in I}(x_e)$ reaches
$(1-\deltabar)|I|$ we double the size of the interval and set
$I':=[\ell-{|I|}/{2},r+{|I|}/{2}]$ (and then choose a new cut-edge
inside $I'$). There is one exception to this growth rule: when the length of
the new interval would be larger than $k+1$ we define
$I':=[\ell-\lceil(k+1-|I|)/2\rceil,r+\lceil(k+1-|I|)/2\rceil]$, i.e., we only
give the new interval a length of $k+1$.

The first change to the hitting game approach is as follows. Whenever an
interval $I$ becomes $\deltabar$-monochromatic (note that this can only
  happen directly after growing the interval) for $\bar{\delta}$ we
stop the growth process for this interval. The interval becomes
\emph{inactive} and we do not maintain a cut-edge inside it, anymore. Consequently,
the cost for the interval will not increase in the future. Note that deleting a cut-edge may induce some cost for the
scheduling algorithm as it has to ensure that the neighbouring slices are
scheduled on the same server.

The second change is that we sometimes deactivate intervals, in order to
guarantee that intervals do not overlap too much. This is done as follows.
After an interval $I$ grows, we deactivate all intervals that are completely
contained in $I$ (we call such intervals \emph{dominated}). This means we stop
the growth process of these intervals and do not maintain a cut-edge for them
anymore. Also this type of deactivation may generate additional cost because
the neighbouring slices may have to be moved. The formal Slicing Procedure is
shown in Algorithm~\ref{alg:cut}.

Note that initially an interval has length $2$ as it contains the two end-points
of an initial cut-edge in the distribution of processes. Right after the first
request to the single edge inside an initial interval the interval will grow.
Therefore, we refer to an interval of length $2$ as an \emph{initial} interval.

We call an interval that has reached its maximum length a \emph{final} interval
as it will not grow anymore. Note that such an interval is always
\emph{active}. An interval may be \emph{inactive} because it is either
$\deltabar$-\emph{monochromatic} or it is \emph{dominated} (completely
contained in another interval). We call two active intervals $I$ and $J$ \emph{adjacent}, if no other active interval has its center between the centers of $I$ and $J$.

\begin{algorithm}[t]
\caption{Slicing procedure}\label{alg:cut}
\KwData{requested edge $e$, set of active intervals $\mathcal{I}$}
\KwResult{new set of active intervals $\mathcal{I}$}
\DontPrintSemicolon
$x(e) \gets x(e) + 1$\;
\ForAll{$I \in \mathcal{I}$ s.t. $e \in I$}{
	update cut-edge according to prob. distribution $\nabla\fmin'(x_{I})$
}

\While{$\exists~I$ such that $\min_{e\in I}x(e) \ge |I|$}{
	grow $I$

	\eIf{$I$ is $\bar{\delta}$-monochromatic}{ 
		$\mathcal{I}  \gets \mathcal{I}\setminus I$ \tcp*{$I$ becomes
                  monochromatic}
	} {
		\ForAll{$J \in \mathcal{I}$ such that $J \subseteq I$}{ 
			$\mathcal{I} \gets \mathcal{I}\setminus J$ \tcp*{$J$
                          becomes dominated}
		}
		choose cut-edge according to prob. distribution $\nabla\fmin'(x_{I})$
	}
}
return $\mathcal{I}$
\end{algorithm}

\subsubsection*{The Clustering Procedure}
The clustering procedure groups slices into clusters such that slices that are
$\frac{3}{4}$-mono-chromatic for the same color are all within the same cluster.
In addition it must ensure that any cluster is not too large because the
scheduling procedure will not split clusters among different servers and, thus,
will not be able to find a good schedule if some clusters are very large.

For every color $c$ the clustering procedure maintains one special cluster,
which we call the color $c$ cluster. A slice $S$ either forms a singleton cluster that only contains
slice $S$, or it belongs to the \emph{color $c$} cluster, where $c$ is the
majority color of $S$. The assignment of a slice to a cluster is done as
follows.

\begin{itemize}
\item Initially all slices consist only of a single color and belong to the
cluster for this color.
\item Whenever a slice $S$ changes (by movement of a cut edge, or by merging a
smaller slice to it\footnote{ties broken arbitrarilly}) we examine the new slice $S'$.
\begin{itemize}
\item
If $S'$ does not have a majority color it becomes a singleton cluster.
\item
If $S'$ is $\frac{3}{4}$-monochromatic for some color $c$ it is assigned to the color
$c$ cluster.
\item Otherwise it is assigned to the cluster of its majority color iff $S$ was
assigned to this cluster. This means that if $S$ and $S'$ have the same majority
color $c$ and $S$ was assigned to the color $c$ cluster we assign $S'$ to it.
Otherwise $S'$ forms a singleton cluster.
\end{itemize}
\end{itemize}

\subsubsection*{The Scheduling Procedure}
\label{section:scheduling_alg}
The scheduling procedure gets the input from the clustering procedure and
maintains an assignment of the clusters to the servers such that the load
on the servers is roughly equal.

Upon a new request the clustering algorithm could change 
existing clusters. Some processes might change their cluster (maybe
deleting a cluster in the process), some clusters might split. This
changes the size of clusters and, hence, the load distribution among the
servers can become imbalanced.

Let $X$ be the maximal size of a cluster and $\vio := \max\{2, X/k\}$. The
scheduling procedure \emph{rebalances} the distribution of clusters among
servers such that there are at most $(\vio+\epsilon)k$ processes on any server,
for an $\epsilon > 0$. This is done as follows.

Assume, that after the execution of the clustering algorithm server $s$ has
load greater than $(\vio+\epsilon)k$. We perform the following rebalancing
procedure. While server $s$ has load greater than $\vio k$, we take
the smallest cluster $C$ in $s$ ($|C|\le D$) and move it to a server $s'$ with load
at most $k$ (such a server must exist because the average load is at most $k$).
If $|C|\le k$ then $s'$ has load at most $2k$ afterwards. Otherwise we migrate the
content of $s'$ (apart from cluster $C$) onto another server with
load at most $k$.

\subsection{Structure of Intervals and Slices}
In order to derive upper bounds on the communication and moving cost of our
algorithm, we have to understand the underlying structure of the intervals and
slices produced in the slicing procedure. 

The next lemma states that monochromatic segments that have a large enough
intersection must have the same majority color.
\begin{restatable}{lemma}{strucone}\label{lemma:overlapping_segments}
	Let $I$ and $J$ be two overlapping $\delta$-monochromatic segments. 
	If  $|I\cap J| \ge \alpha \max \{|I|, |J| \}$ and $\delta \ge 1-\frac{\alpha}{2}$, the segments have the same majority color.  
\end{restatable}
\begin{proof}
Let $|I\setminus J|=a$, $|I\cap J|=b$ and $|J\setminus I|=c$. Then,
$|I \cup J|=a+b+c$. 
Assume that the segments have
different majority colors $c_I$ and $c_J$, respectively. Then, the segment
$|I \cup J|$ must contain strictly more than $\delta (a+b)$ elements of color $c_I$ and strictly more than
$\delta (b+c)$ elements of color $c_J$. Trivially, the number of elements of
both colors cannot exceed $a+b+c$, i.e.,
$\delta(a+b) + \delta (b+c) < a + b+ c$ or 
$\delta < \frac{a+b+c}{a+2b+c}$. Furthermore, since $b \ge \alpha(a+b)$ and
$b \ge \alpha(b+c)$ we get $b \ge \frac{\alpha}{2}(a+2b+c)$. Then,
\begin{equation*}
	\delta < \frac{a+b+c}{a+2b+c} = 1-  \frac{b}{a+2b+c} \le 1 - \frac{\alpha}{2}\enspace,
\end{equation*}
which gives a contradiction.
\end{proof}

Next, we argue, that the union of a sequence of consecutive overlapping
$\delta$-mono\-chromatic intervals of the same majority color is also monochromatic. 
\begin{restatable}{lemma}{structwo}\label{lemma:overlapping_segments_delta_prime}
Let $I_1, I_2$, $\dots, I_m$ be $\delta$-monochromatic intervals with the same
majority color, such that $I=\bigcup_{i} I_i$ forms a single contiguous
segment. Then $I$ is $\frac{\delta}{2-\delta}$-monochromatic for $c$.
\end{restatable}
\begin{proof}
Wlog.\ we assume that the set of intervals does not contain
\emph{redundant} intervals, i.e., for all $j$, $\bigcup_{i \neq j} I_i \neq I$.
Further, we assume that the intervals are numbered by the order of their centers.

Define $I_0=I_{m+1}=\emptyset$ and let  $X_i=I_{i}\cap I_{i+1}$ and $A_i=I_{i}\setminus X_{i-1}$. 
Then, the intervals have the following structure:
$I_{i}=X_{i-1} \cup A_i \cup X_{i+1}$. Notice, that $X_i$ and $A_j$ are disjoint.
For a segment $S$ we use $f(S)$ to denote 
the number of elements with color $c$ in $S$. Trivially, $0\le f(S)\le |S|$. Furthermore, we know that
\begin{equation*}
	f(I_i)=f(X_{i-1}) + f(A_{i}) + f(X_{i}) \ge \delta(|X_{i-1}| + |S_i| + |X_i|)=\delta |I_i|,
\end{equation*} for all $1\le i \le m$. Then,
\def\dsum{\textstyle\sum\nolimits}
\newlength{\myskip}
\begin{equation*}
\begin{split}
	\delta|I| &= \delta \dsum_{i=1}^m (|X_{i-1}|+ |A_i|)                                                                     \\[\myskip] 
	          &=  \dsum_{i=1}^m \delta(|X_{i-1}| + |A_i| + |X_{i}|) - \dsum_{i=1}^m \delta |X_{i}|                           \\[\myskip] 
	          &\le \dsum_{i=1}^m (f(X_{i-1}) + f(A_{i}) + f(X_{i})) - \dsum_{i=1}^m \delta |X_{i}|                           \\[\myskip] 
	          &= \dsum_{i=1}^m (f(X_{i}) + f(A_{i})) + \dsum_{i=1}^m (f(X_{i-1}) - \delta |X_{i}|)                           \\[\myskip] 
	          &= \dsum_{i=1}^m (f(X_{i}) + f(A_{i})) + \dsum_{i=1}^m (f(X_{i}) - \delta |X_{i}|)                             \\[\myskip] 
	          &= \dsum_{i=1}^m (f(X_{i}) + f(A_{i})) + \dsum_{i=1}^m ((1-\delta)f(X_{i})  + \delta f(X_{i}) - \delta |X_{i}|)\\[\myskip] 
	          &\le \dsum_{i=1}^m (f(X_{i}) + f(A_{i})) + (1-\delta)\dsum_{i=1}^m f(X_{i})                                    \\[\myskip] 
	          &\le (2-\delta)\dsum_{i=1}^m (f(X_{i}) + f(A_{i}))                                                             \\[\myskip] 
	          &= (2-\delta)f(I) \qedhere
\end{split}                
\end{equation*}%
\end{proof}

\newlength{\enumerateparindent}
\setlength{\enumerateparindent}{0pt}
\def\enumindent{\hspace*{0pt}\hspace{\enumerateparindent}}

The next lemma explains the structure of the slice between two adjacent active intervals.
Let $I$ be an interval with length $|I|\ge 4$ and let $I'$ denote the interval
before its most recent growth step.
We say $I'$ is the core of $I$ and denote it by $\core(I)=I'$. Note that
initial intervals do not have a core.
\begin{restatable}{lemma}{structhree}\label{lemma:segment_structure}
Let $N=[a,b]$ be a segment such that no active interval is intersecting $N$.
Let $\mathcal{M}$ be the set of inactive intervals $I$ with
$\CORE(I) \cap N \neq \emptyset$. Let $M=\bigcup_{I\in \mathcal{M}} I$ and
$F=N \setminus M$.
Then, there exists a color $c$ and an interval set $\mathcal{U}\subseteq \mathcal{M}$ with
the following properties:
\begin{enumerate}[a)]
	\item $\bigcup_{I\in \mathcal{U}} I=M$;\label{pro:cover}
	\item each interval in $\mathcal{U}$ is $\deltabar$-monochromatic for
        color $c$;\label{pro:mono}
	\item each process in $F$ has initial color $c$;\label{pro:free}
	\item the segment $N'= M \cup F$ is
          $\dbar/(2-\dbar)$-monochromatic for color $c$.\label{pro:strange}
\end{enumerate}
\end{restatable}
\begin{proof}
Property~\ref{pro:strange} just follows from the first three properties by
simply applying Lemma~\ref{lemma:overlapping_segments_delta_prime}. Therefore,
we focus on proving properties \ref{pro:cover}, \ref{pro:mono} and \ref{pro:free}.
We prove
these by induction over the number of intervals in $\mathcal{M}$.

\medskip
\noindent
\textit{Base case $(|\mathcal{M}|=0)$:}\hfill\\
If $\mathcal{M}=\emptyset$ we have $F=N\setminus M=N$ and
$\mathcal{U}=\emptyset$. Then properties \ref{pro:cover} and \ref{pro:mono} of the lemma are trivially
fulfilled. Any initial cut-edge is either inside an active interval or inside the core of some
interval. As no active interval and no core intersects $N$ (otw.\ $\mathcal{M}$ would not be empty) we
get that all processes in $F$ must have the same initial color. We choose $c$
as this color. This fulfills Property~\ref{pro:free}.

\medskip
\noindent
\textit{Induction step:}\hfill\\ We call the vertices in $F$ \emph{free vertices} as they are
not contained in any set of $\mathcal{M}$. We distinguish two cases.

\medskip
\begin{enumerate}[I)]\parskip0pt\enumerateparindent2ex
\item
First suppose that there exists a free vertex $f$ that has an interval of
$\mathcal{M}$ on both sides. Let $\mathcal{M}_\ell$ denote the set of interval to the left
of $f$ and $\mathcal{M}_r$ denote the interval to the right of $f$. We apply
the induction hypothesis on the sub-sequences $N_\ell:=[a,f]$ and $N_r:=[f,b]$
(with interval sets $\mathcal{M}_\ell$ and $\mathcal{M}_r$, respectively). This
gives interval sets $\mathcal{U}_\ell$ and $\mathcal{U}_r$ that are
$\deltabar$-monochromatic for colors $c_\ell$ and $c_r$, respectively, and it
gives sets of free vertices $F_\ell$ and $F_r$ that all have color $c_\ell$ and $c_r$,
respectively.

\enumindent
The union $\mathcal{U}:=\mathcal{U}_\ell\cup\mathcal{U}_r$ of the interval sets
covers $M$, which gives Property~\ref{pro:cover}. The free vertices are
$F=F_\ell\cup F_r$. To show properties~\ref{pro:mono} and \ref{pro:free} we
have to argue that $c_\ell=c_r$. But this is immediate because by construction
$F_\ell\cap F_r\neq\emptyset$ as it contains vertex $f$.

\item
Now suppose that $M=\cup_{i\in I}I$ forms a single contiguous segment. 
Let $H=[h_\ell, h_r]$ with $\CORE(H)=[c_\ell, c_r]$ be the largest interval in
$\mathcal{M}$ (ties broken arbitrarily). Let $N_\ell=[a,c_\ell]$ and
$N_r=[c_r,b]$ denote the segments to the left and right of $H$'s core (but
still sharing one vertex with the core),
and let $F_\ell:=N_\ell-M$ and $F_r:=N_r-M$ denote the free
vertices in these segments.
Further define
$\mathcal{M}_\ell=\{I\in \mathcal{M}~|~\text{s.t.}~\CORE(I)\cap N_l \neq \emptyset\}$,
$\mathcal{M}_r=\{I\in \mathcal{M}~|~\text{s.t.}~\CORE(I)\cap S_r \neq \emptyset\}$,
$M_\ell = \bigcup_{I\in \mathcal{M}_\ell} I$, and
$M_r = \bigcup_{I\in \mathcal{M}_r} I$. 

\enumindent
By construction $M=M_\ell\cup M_r$.
We now proceed by constructing
interval sets $\mathcal{U}_\ell\subseteq
\mathcal{M}_\ell$ and $\mathcal{U}_r\subseteq \mathcal{M}_r$ that cover
$M_\ell$ and $M_r$, respectively, and that are $\deltabar$-monochromatic
for colors $c_\ell$ and $c_r$, respectively.
This is sufficient to obtain properties \ref{pro:cover} and \ref{pro:mono}
by choosing $\mathcal{U}=\mathcal{U}_\ell\cup\mathcal{U}_r$ because
$c_\ell$ will be equal to $c_r$. To see this we argue that both $\mathcal{U}_\ell$
and $\mathcal{U}_r$ must contain the set $H$. This holds for $\mathcal{U}_r$
because $H$ is the only set in $\mathcal{M}_r$ that contains 
the vertex $h_\ell\in M_r$. An interval $X\in\mathcal{M}_r$ has a vertex to the right of $c_r$ (including
$c_r$) in its core. To also include $h_\ell$ would imply $H\subseteq X$. This is
not possible as $H$ is one of the largest intervals and no two intervals can
have identical borders.

\enumindent
Now, we show how to construct 
$\mathcal{U}_r$. The construction for
$\mathcal{U}_\ell$ is analogous. We distinguish two sub-cases.

\begin{enumerate}[A)]
\item
First suppose that $h_r$ is the rightmost vertex
in any interval of $\mathcal{M}$.
Then any interval in $\mathcal{M}_r$ is completely contained in $H$, which means we can
choose $\mathcal{U}_r=\{H\}$.

\item
Now, let $J\neq H$ denote the interval that contains the rightmost vertex of
$M$. First observe that $J$ is $\deltabar$-monochromatic because if not there
would need to be a strictly larger $\deltabar$-monochromatic interval that
completely contains $J$ in its core. Then $J$ could not
contain the rightmost vertex of $M$.

\begin{enumerate}[a)]\parskip0pt\enumerateparindent2ex
\item
If $H$ and $J$ have intersecting cores then the whole segment $[c_r,h_r]$ is
contained in both of them. This means that their intersection is at least
$\tfrac{1}{4}\max\{|H|,|J|\}$. Since $\dbar/(2-\dbar) \ge \frac{7}{8}$ we can
apply Lemma~\ref{lemma:overlapping_segments} and conclude that $H$ and $J$ have
the same majority color. Hence, we can choose
$\mathcal{U}_r=\{H,J\}$ to cover $M_r$.
\item
Suppose the cores do not intersect. Let $J=[j_\ell,j_r]$ and \CORE(j)=$[d_\ell,d_r]$.
We apply the induction hypothesis to the two
sub-sequences $A=[c_r,d_\ell-1]$ and $B=[c_r+1,d_\ell]$. We can do this because
both have at most $|\mathcal{M}|-1$ intersecting cores. We get two
$\dbar/(2-\dbar)$-monochromatic sequences $A'$ and $B'$, and interval subsets
$\mathcal{U_A}$ and $\mathcal{U_B}$.

\enumindent
$A'=[h_\ell,\gamma]$,
where $\gamma\ge\max\{h_r,d_\ell-1\}$ because $A'$ must cover $H$
and $A$.
$B'=[\gamma,j_r]$, where $\gamma\le\min\{c_\ell+1,j_\ell\}$ because $B'$ must
cover $J$ and $B$.

\enumindent
The intersection between $A'$ and $B'$ has size at least
$\frac{1}{4}\max\{|A'|,|B'|\}$, which by Lemma~\ref{lemma:overlapping_segments}
gives that they need to have the same majority color. Also the intervals in
$\mathcal{U}_A$ and $\mathcal{U_B}$ have this majority color. Hence, we can
choose $\mathcal{U}_r=\mathcal{U}_A\cup\mathcal{U}_B$.

\end{enumerate}

\end{enumerate}

So far all intervals in $\mathcal{U}=\mathcal{U}_\ell\cup\mathcal{U}_r$ have
the same majority color $c$. It remains to show that also all in $F_\ell$
and $F_r$ have this color. This then gives Property \ref{pro:free} of the lemma.
Again we only do the argument for $F_r$.

\enumindent
Let $X=[x_\ell,x_r]$ denote the interval that contains the rightmost vertex of $M$. 
If $b\ge x_r$ then $F_r=\emptyset$ and there is nothing to prove. Otherwise,
$F_r=[x_r+1,b]$. We first argue that
the leftmost vertex $x_r+1$ in $F_r$ has color $c$. This implies the statement
because the region $[x_r+1,b]$ cannot contain any initial cut-edges as these
would be either inside an active interval or 
inside the core of an inactive interval, and there are no cores intersecting
$[x_r+1,b]$.

\enumindent
We apply the induction hypothesis to the segment $S:=[c_r+1,h_r+1]$, which we
can do because $\CORE(H)$ does not intersect $S$ (hence, at most
$|\mathcal{M}|-1$ intervals have a core intersecting $S$). By
Property~\ref{pro:strange} this returns a segment $S'=M_r\cup F_S\supseteq S$
that is $\deltabar/(2-\deltabar)$-monochromatic for some color $c_S$, and all
vertices in $F_S=S\setminus M_r$ have color $c_S$.

\enumindent
The intersection between $S'$ and $H$ is at least $\tfrac{1}{4}\max\{|H|,|S'|\}$.
Since $\dbar/(2-\dbar) \ge \frac{7}{8}$ we can apply Lemma~\ref{lemma:overlapping_segments} and
conclude that $H$ and $S'$ have the same majority color, i.e., $c_S=c$.
In addition we get that the vertex $x_r+1$ has color $c$, as well, as it is
contained in $F_S$.

\enumindent
This finishes the proof of Case~II of the induction step. \qedhere
\end{enumerate}
\end{proof}
\begin{restatable}{lemma}{strucfour}\label{lemma:gamma_mono}
Let $L=[a,b]$ and $R=[c,d]$ be two non-intersecting intervals, such that no active interval intersects the segment $N=[b+1, c-1]$. Let
$e_\ell=(\ell, \ell+1)$ and $e_r=(r, r+1)$ be the centers of $L$ and $R$, respectively, and let $S=[a, d]$ be
the segment containing both $L$ and $R$. Let $\mathcal{M}$ be the set of
non-initial intervals with centers in $[\ell+1, r]$ and let $M = \sum_{I\in \mathcal{M}} |I|$.

If $|S| \ge \frac{1}{1-\delta}(M + |L| + |R|)$, then each segment
$S'=[i, j]$ with $i\in L$ and $j \in R$ is $\delta$-monochromatic.
\end{restatable}
\begin{proof}
We show that if $|S| \ge \frac{1}{1-\dbar}(M + |L| + |R|)$, then the segment $N$ contains many processes with same initial color, which implies the lemma, since every segment $S'$ contains $N$.

Let $\mathcal{M'}$ be the set of intervals $I$ such that $\core(I) \cap N \neq \emptyset$. Then, $\mathcal{M'} \subseteq  \mathcal{M}$. Otherwise, w.l.o.g. there exists an interval $I$ with center $e_i < e_\ell$ and $\core(I) \cap N \neq \emptyset$. But, then  $L \subseteq \core(I)$, which means $L$ should have been dominated at the time when $I$ was active.

We apply Lemma~\ref{lemma:segment_structure} on $N$ and receive a set of $\dbar$-monochromatic (non-initial) intervals $\mathcal{U} \subseteq \mathcal{M'}$, such that $U=\bigcup_{I\in \mathcal{U}}I$ and the processes in $F=N \setminus U$ have the same initial color. Then, 
\[|F| = |N| - |U| = |S| - (|U| + |L| + |R|)\enspace.\]
Since $\mathcal{U} \subseteq \mathcal{M'} \subseteq  \mathcal{M}$, and therefore $|U| \le M$, we conclude:
\[
|F| \ge |N| -  (M + |L| + |R|)
\ge |S| - (1-\delta) |S|
= \delta |S|
\ge \delta |S'|
\]
Then, there are at least $\delta |S'|$ processes of the same color which means that segment $S'$ is $\delta$-monochromatic. 
\end{proof}

\subsection{Correctness}

In this section we show that at any time step, each server contains at most $(3+2\epsilon)k$ processes. A crucial ingredient for this is, that each cluster has a bounded size. 

\begin{restatable}{lemma}{correctone}
Let $A$ and $B$ be two adjacent active intervals with cut-edges $e_a$ and
$e_b$, respectively. The size of the slice $S$ between $e_a$ and $e_b$ is at
most $|A| + |B| -2 + (2-\dbar)/{\dbar}\cdot k$.
\end{restatable}
\begin{proof}
W.l.o.g. let $A=[a_\ell,a_r]$ and $B=[b_\ell,b_r]$ with $0\le a_\ell< b_\ell<n$.
If the intervals intersect ($a_r > b_\ell$) we are done.
Otherwise we analyze the segment $N = [a_r+1, b_\ell-1]$.
Since $N$ does not intersect any active interval we can apply
Lemma~\ref{lemma:segment_structure}, which states, that there is a $\frac{\dbar}{2-\dbar}$-monochromatic segment $N'$ with $N \subseteq 
N'$. 
Thus, $N'$ consists of at most $k$ processes of some color $c$ and at most
$(1-\frac{\dbar}{2-\dbar})k$ processes of other color. 
Hence, $|N|\le |N'| \le  k +(1-\frac{\dbar}{2-\dbar})|N'| \le\frac{2-\dbar}{\dbar}k$.
Then, the slice between $e_a$ and $e_b$ contains $N$, at most $|A|-1$ processes of $A$ and at most $|B|-1$ processes of $B$, which is at most $|A| + |B| -2 + \frac{2-\dbar}{\dbar}k$.
\end{proof}

\begin{corollary}\label{lemma:singleton_slice_size}
A singleton cluster has size at most $(3+{2(1-\dbar)}/{\dbar})\cdot k$.
\end{corollary}
\begin{proof}
Since the size of each interval is at most $k+1$, we get that each singleton slice has size at most $2k +{(2-\dbar)}/{\dbar}\cdot k = (3 + 2(1-\dbar)/{\dbar})\cdot k$.
\end{proof}

\begin{observation}\label{lemma:c-colored}
	Let $S$ and $T$ be $\frac{3}{4}$-monochromatic segments with majority color $c$. Then, $S$ and $T$ are both contained in the color $c$ cluster.
\end{observation}

\begin{lemma}\label{lemma:colored_slice_size}
The size of a color $c$ cluster is at most $2k$.
\end{lemma}
\begin{proof}
A color $c$ cluster contains at most $k$ processes of color $c$. Since the slices in this cluster are all at least ${1}/{2}$-monochromatic, the number of processes of other colors in this cluster is at most $k$. Thus, the overall number of processes in a color $c$ cluster is at most $2k$.
\end{proof}

\begin{lemma}\label{lemma:correctness_scheduling}
The scheduling algorithm produces an assignment of processes to servers, such that the load of each server is at most $(3+2\epsilon)k$, for $\dbar \ge {2}/(2+\epsilon)$.
\end{lemma}
\begin{proof}
By construction, the scheduling algorithm ensures that no server has load greater that $(\vio+\epsilon)k$, where $\vio k$ is the maximum slice size. Due to Corollary~\ref{lemma:singleton_slice_size} and Lemma~\ref{lemma:colored_slice_size} we know that $\vio \le 3 + {2(1-\dbar)}/{\dbar}$, which is at most $3+\epsilon$ for $\dbar \ge {2}/(2+\epsilon)$.
\end{proof}

\subsection{Cost Analysis}

\subsubsection{Cost of Intervals}
\def\cost{\operatorname{cost}}
\def\costhit{\operatorname{cost}_{\operatorname{hit}}}
\def\costmove{\operatorname{cost}_{\operatorname{move}}}

We view parts of the cost of the online algorithm as associated with intervals,
as follows.
For an interval $I$ we use $\costhit(I)$ to denote the communication cost that the online
algorithm experiences on the cut-edge maintained by the interval, 
we use $\costmove(I)$ to denote the cost for moving the cut-edge within the interval, and
define $\cost(I):=\costhit(I)+\costmove(I)$ to be the \emph{cost of the
  interval}. Observe that the cost of merging neighboring slices when deactivating
an interval is not counted in $\cost(I)$.

\begin{observation}
For an initial interval $\cost(I)=0$.
\end{observation}

We use $\OPT(I)$ to denote the cost that the optimum algorithm
experiences for migrating processes of $I$, or for communicating along edges in
$I$.

\begin{lemma}\label{lem:optlower}
We have the following lower bounds on the optimum cost of an interval:
\begin{enumerate}
\item A non-initial interval $I$ fulfills $\OPT(I)\ge \frac{1}{2}(1-\deltabar)|I|$.
\item Active intervals fulfill  $\OPT(I)\ge \frac{1}{2}\min_ex_e$, where $x_e$
is the current request vector.
\item Inactive intervals fulfill $\OPT(I)\ge \frac{1}{2}\min_ex_e'$, where
$x_e'$ is the request vector at the time $I$ became inactive.
\end{enumerate}
\end{lemma}

\begin{proof}
For the first part $I$ is non-initial, which means that it has grown before. Let $I'$ denote the
interval $I$ before the most recent growth step. The interval $I'$ was grown
because we had $\min_{e\in I'}x_e\ge(1-\deltabar)|I'|$. This means if $\OPT$ has
a cut edge inside $I'$ it experiences cost at least $\min_{e\in
  I'}x_e\ge(1-\deltabar)|I'|$; if not it had to recolor at least
$(1-\deltabar)|I'|$ processes in $I'$ because the interval $I'$ is not
$\deltabar$-monochromatic. Consequently we get
\begin{equation}
\label{eqn:obs}
\OPT(I)\ge\OPT(I')\ge (1-\deltabar)|I'|\ge\tfrac{1}{2}(1-\deltabar)|I| \enspace,
\end{equation}
where the last step follows because the growth step at most doubles the length
of an interval.

The second part trivially holds for initial intervals, as the left hand side is
$0$ in this case.
Intervals that are non-final (i.e., they did not reach the maximum size
$k+1$ yet), fulfill $\min_{e\in I}x_e\le (1-\deltabar)|I|$, as otw.\ we would grow
the interval. Then Equation~\ref{eqn:obs} directly implies the lemma. Final
intervals have size $k+1$ and therefore the optimum algorithm must have a
cut-edge inside. Hence, $\OPT(I) \ge \min_{e\in I}x_e$ for such intervals.
The third part immediately follows from Part~2.
\end{proof}

\begin{lemma} \label{lemma:interval_cost}
For any interval
$\cost(I) \le \mcO({1}/(1-\dbar)\cdot\log k) \OPT(I)$.
\end{lemma}
\begin{proof}
Clearly, the lemma holds, if $I$ is an initial interval, since $\cost(I) = 0$.
Let $I$ be a non-initial interval and let $x$ be either the current request vector (if $I$ is still active) or the request vector at the time $I$ was deactivated. 
Due to Lemma~\ref{lemma:grow_alg} we know that $I$ has cost at most
$E[\cost(I)] \le \mcO(1)\min_e(x_e) + \mathcal{O}(\log |I|)~|I|$.
Using Lemma~\ref{lem:optlower} we derive that $E[\cost(I)] \le \mcO(\frac{1}{1-\dbar}\log k) \OPT(I)$.
\end{proof}

\subsubsection{Cost of the Algorithm}
The overall cost of the algorithm consist of five parts. The first part is
the \emph{communication cost}, which we denote by $\costhit$.

The next three parts are migration costs due to the clustering algorithm.
Whenever a process is moved to an existing cluster the scheduling algorithm \emph{may}
have to also move this process in order to guarantee that all processes in a
group/cluster are scheduled on the same server. Therefore we define the following costs.

\begin{itemize}
\item \emph{Moving cost} $\costmove$: when an interval moves a cut-edge by a distance
of $d$, $d$ processes switch the slice that they belong to. This increases the
\emph{moving cost} of the clustering algorithm by $d$.
\item \emph{Merging cost} $\CMerge$: when two slices $S_s$ and $S_\ell$ with
$|S_s|\le |S_\ell|$ are merged the processes in the smaller slice are moved
to the cluster of the larger slice, which results in a cost of at most $|S_s|$.
\item \emph{Monochromatic cost} $\CMono$: when a slice $S$ becomes
$\frac{3}{4}$-monochromatic for some color $c$ (and it was not
$\frac{3}{4}$-mono\-chromatic before) we move it to the color $c$ cluster at a
cost of $|S|$.
\end{itemize}
Observe that there is no cost if we remove a slice from a color $c$ cluster to
form a singleton cluster, as this does not put a constraint on the scheduling
algorithm to move any processes.

\def\costbal{\ensuremath{\operatorname{cost}_{\operatorname{bal}}}}
The final part of the cost is the migration cost caused by direct actions of
the scheduling procedure to enforce capacity constraints. We call this the
\emph{rebalancing cost} $\costbal$.

%
%

\begin{observation}\label{obs:hit_move_costs}
The communication cost and moving cost are given by the interval cost, this means 
$\costhit = \sum_I\costhit(I)$, and $\costmove = \sum_I\costmove(I)$
\end{observation}



\begin{restatable}{lemma}{mergelemma}\label{lemma:merge_cost}
The merge cost $\CMerge$ of the clustering algorithm is at most $\mcO({1}/{(1-\dbar)}\cdot\log k)\sum_{I \in \mathcal{I}} \OPT(I)$.
\end{restatable}
\begin{proof}
\def\costmig{\operatorname{cost}_{\operatorname{merge}}}
For each non-initial interval $I$ we create an account $\costmig(I)$.
We show how to distribute the merge costs among these accounts such that $$\CMerge \le \sum_
{I \in \mathcal{I}} \costmig(I) \le \sum_
{I \in \mathcal{I}} \mcO(\log k) |I|\enspace.$$ Since for every non-initial interval $\OPT(I) \ge \frac{1}{2}(1-\dbar) |I|$ according to Lemma~\ref{lem:optlower}, the lemma follows. 

Merging costs arise only if we deactivate an interval and remove its cut-edge. This could happen in two cases. 

\begin{enumerate}[I)]
\parskip0pt\enumerateparindent2ex
\item 
First, during a growth step of an interval $I$, some intervals might become dominated by $I$, i.e., they are completely inside $I$. We deactivate these intervals and remove their corresponding cut edges. However, an interval that gets dominated by $I$ must have its cut edge inside $I$. But then, the cost for merging the slices due to the deactivation of dominated (by $I$) intervals is at most $|I|$. We charge $I$ the cost $|I|$. Let $I_1, I_2, \dots, I_m$ denote the growth phases of $I$, i.e., $I_m = I$ and $|I_i| = 2^i$ (unless of course $I$ is a final interval, where $|I|=|I_m|=k+1$). We derive, that an interval is charged for Case~I at most $\sum_i |I_i| \le \mcO(1) |I|$.
\item
Second, during a growth step of an interval $J$, it may become $\dbar$-mono\-chromatic. We deactivate $J$ and remove its cut-edge. We show how to distribute the cost for the slice merge.

\enumindent
Let $L=[a,b]$, $J=[c,d]$ and $R=[e,f]$ be adjacent active intervals with cut edges $e_\ell$, $e_j$ and edges $e_r$, respectively. Let $S'$ be the slice between $e_\ell$ and $e_j$, and $T'$ the slice between $e_j$ and $e_r$. Furthermore, let $S=[a,d]$ and $T=[c,f]$. Clearly, $|S'| \le |S|$ and $|T'| \le |T|$.
Assume interval $J$ becomes $\dbar$-monochromatic, i.e., we have to eliminate the cut-edge $e_j$ and merge the incident slices $S'$ and $T'$. 
If both $S'$ and $T'$ are $\frac{3}{4}$-monochromatic, then due to Observation~\ref{lemma:c-colored} the slices already reside in the same cluster, such that no cost is incurred. 
Otherwise, we move the smaller slice to the cluster of the bigger slice and pay $X=\min \{|S'|, |T'|\}$. 
We assume that either $S'$ or $T'$ is not $\frac{3}{4}$-monochromatic. 

\enumindent
We introduce additional notation.
For adjacent active intervals $I$ and $J$ with centers $e_i=(i, i+1)$ and $e_j=(j, j+1)$, respectively, we define the \emph{center slice} between $I$ and $J$ as the segment $C=[i+1, j]$. 
We say that the active intervals $I$ and $J$ are \emph{incident} to the center slice $C$. 
Furthermore, we say an inactive interval $I$ with its center in $C$ is also incident to $C$.
Notice, that an inactive interval is incident to exactly one center slice and an active interval is incident to exactly two center slices. 

\enumindent
Let $S_c$ be the center slice between $L$ and $J$, and let $T_c$ be the center slice between $J$ and $R$, respectively. We distribute the cost $X$ to the intervals incident to the center slices $S_c$ and $T_c$.

\enumindent
We show that each interval $I$ is charged at most $\mcO(\log k)|I|$. The basic argument is that whenever an interval $I$ is charged a cost $\mcO(|I|)$, one of its incident center slices doubles its size. 
Observe, that due to Corollary~\ref{lemma:singleton_slice_size} and Lemma~\ref{lemma:colored_slice_size} the size of a center slice is at most $\mcO(k)$. 
Furthermore, an interval is incident to at most two center slices. Hence, if an interval is charged $\mcO(\log k)|I|$ merge cost, then its incident center slices must reach the maximum possible size. 

\enumindent
Observe, that $J$ cannot be an initial interval, since it became $\dbar$-monochromatic. However, if $L$ or $R$ are initial intervals, we cannot charge them. Instead we charge $J$ the additional cost (which is at most constant).

\enumindent
Assume, w.l.o.g that $|S_c| \le |T_c|$.
We distinguish between two cases: 

\begin{enumerate}[A)]
\item Assume that $S'$ is not $\frac{3}{4}$-monochromatic or $L\cap J \neq \emptyset$.
Let $\mathcal{M}_S$ be the set of non-initial intervals with centers in $S$ and let $M=\sum_{I\in \mathcal{M}_S}|I|$.
Clearly, $X \le |S'| \le |S|$.
If $L\cap J \neq \emptyset$, then $X \le |S| \le |L| + |J|$.
If $S'$ is not $\frac{3}{4}$-monochromatic and $L\cap J = \emptyset$, then we apply Lemma~\ref{lemma:gamma_mono} on $L$ and $J$ and derive that $X \le |S| < 4(M +|J| +|L|)$. 

We distribute the cost $X$ among the intervals in $\mathcal{M}_S \cup \{J, L\}$.
We charge each interval $I\in \mathcal{M}_S \cup \{J, L\}$ the cost $4|I|$. 
If $L$ is an initial interval, then we charge instead $J$ the additional constant cost.
Then, $4(M + |J| +|L|) \ge X$, i.e., we distributed all cost $X$. 

Now, we fix an interval $I\in \mathcal{M}_S\cup\{J,L\}$ and assume that $I$ does not grow during future time steps (i.e., we analyze a fixed growth phase of an interval).
After the deactivation of $J$, the interval $I  \in \mathcal{M}_S \cup \{J, L\}$ is incident to the center slice $Z=T_c \cup S_c$ with $|Z| \ge 2 |S_c|$. Thus, for $I$ one of its incident center slices at least doubles its size. Hence, if we charge $I$ $\mcO(\log k)$ times (each time $|I|$) for Case~A, its incident center slices reach their maximum size. 
We conclude that $I$ pays at most $\mcO(\log k) |I|$ for Case~A. 
\item Now, assume that $S'$ is $\frac{3}{4}$-monochromatic and $L\cap J = \emptyset$.
Since $S'$ is $\frac{3}{4}$-monochromatic, then $T'$ must be non-$\frac{3}{4}$-monochromatic.
Let $\mathcal{M}_T$ be the set of non-initial intervals with centers in $T$ and let $M=\sum_{I\in \mathcal{M}_T}|I|$.
If $J \cap R \neq \emptyset$, then clearly, $X \le |T| \le |J| + |R|$. Otherwise, we apply Lemma~\ref{lemma:gamma_mono} on $J$ and $R$ and derive that $X \le |T| < 4(M + |R| +|J|)$. 

We distribute the cost $X$ among the intervals in $\mathcal{M}_T \cup \{J,R\}$. 
We charge each interval $I\in \mathcal{M}_T\cup\{J,R\}$ the cost $4\frac{|I|}{|T|}X$. If $R$ is an initial interval, we charge instead $J$ the additional constant cost. Then, $\sum_{I}4\frac{|I|}{|T|}X = \frac{4}{|T|}(M +|J| +|R|)X \ge X$, i.e., we distributed all cost $X$.

Now, we fix an interval $I'\in \mathcal{M}_T\cup\{J,R\}$ and assume that $I'$ does not grow during future time steps (i.e., we analyze a fixed growth phase of an interval).
After the deactivation of $J$, $I'$ is incident to the center slice $Z=T_c \cup S_c$. 
Since $L\cap J = \emptyset$, then $X \le |S| \le 2|S_c|$, i.e., the center slice $Z$ increases its size by at least $\frac{1}{2}X$. 
Now, assume that $I'$ is charged repeatedly $4\frac{|I'|}{|T_1|}X_1, 4\frac{|I'|}{|T_2|}X_2, \dots $ for Case~B. Clearly, $|T| = |T_1| \le |T_2| \le ...$, since due to merges the segment $T$ only grows. Assume, that at some point $I'$ experiences at least $8|I'|$ cost due to Case~B. Then, $8|I'| \le \sum_i 4\frac{|I'|}{|T_i|} X_i \le 8\frac{|I'|}{|T|} \sum_i \frac{1}{2} X_i$, i.e., $|T| \le \sum_i \frac{1}{2} X_i := Y$. Then, we know that the center slice $Z$ grows by at least $Y \ge |T| \ge |T_c|$, which means that $|Z| \ge 2|T_c|$, i.e., the size of the center slice $Z$ doubles, while we charge interval $I'$ at most $\mcO(1)|I'|$. 
Hence, if we charge $I'$ overall $\mcO(\log k)|I'|$ merging cost, its incident center slices reach their maximum size. 
We conclude, that interval $I'$ pays for Case~B at most $\mcO(\log k)|I'|$.
\end{enumerate}

Let $I$ be an interval and let $I_1, I_2, \dots, I_m$ denote the growth phases of $I$, i.e., $I_m = I$ and $|I_i| = 2^i$ (unless of course $I$ is a final interval, where $|I|=|I_m|=k+1$). From the above two 
cases (A and B) we derive that during each growth phase $I_i$ is charged at most $
\mcO(\log k)|I_i|$. Thus, the overall cost that we charge $I$ is at most $\costmerge(I) \le \sum_i \mcO(\log k)|I_i| \le \mcO(\log k) |I|$.
\end{enumerate}

Summarizing over all cases we conclude that for each non-initial interval $I$ the merge cost $\costmerge(I)$ is at most $\mcO(\log k) |I|$.
\end{proof}

\begin{restatable}{lemma}{monocost}\label{lemma:col_cost}
The monochromatic cost $\CMono$ of the clustering algorithm is at most $\mathcal{O}(1) (\CMove + \CMerge)$.
\end{restatable}
\begin{proof}
At the beginning, every slice is a $1$-monochromatic segment and therefore contained in a colored cluster. Assume, at some point a slice $S$ becomes $\frac{3}{4}$-monochromatic. For $S$ to leave a colored cluster and become $\frac{3}{4}$-monochromatic again, $S$ should have experienced at least $(\frac{3}{4}-\frac{1}{2})|S|$ change in size (due to moves and merges) since it left a colored cluster. The cost for migrating $S$ to the colored cluster is $|S|$. Then, the cost for all "colored" migrations is at most $O(1)$ times the overall moving and merge cost.
\end{proof}

\begin{lemma}\label{lemma:cost_migrations}
The rebalancing cost  $\costbal$ of the algorithm is at most $
\mcO(\frac{1}{\epsilon})(\costmove +$ $\costmerge + \costmono)$.
\end{lemma}
\begin{proof}
Initially, each server has load $k$.
Assume, that after the clustering step server $s$ has load $L$ with $L > (\vio+\epsilon)k$, i.e., server $s$ becomes imbalanced. The costs of the cluster migrations due to the rebalancing procedure are at most $L$. We know, that after the previous rebalancing step for server $s$ the corresponding server had load at most $\vio$. Hence, in the time between the previous and the current rebalancing, server $s$ observed at least $L-\vio$ migrations. 
These migrations were due to the clustering algorithm, since the rebalancing procedure do not place more than $\vio$ processes on a server. 
Since $\frac{L}{L-\vio} \le \frac{\vio + \epsilon}{\epsilon} \le \frac{\vio+1}{\epsilon}$ (for $\epsilon \le 1$), we conclude that the cost of the rebalancing procedure is at most $\mcO(\frac{1}{\epsilon}) (\costmove + \costmerge + \costmono)$ (since $D$ is constant). 
\end{proof}

\subsubsection*{Lower Bound}
Now, we show a lower bound on the cost of an optimum algorithm. Lemma~\ref{lem:optlower}
already gives us a lower bound for the cost of $\OPT$ on an interval $I$. The
following lemma shows that a process is not contained in too many intervals,
which allows us to extend the lower bound.

\begin{lemma}\label{lemma:contained_in_log}
A process $p$ is contained in at most $\mathcal{O}(\log k)$ intervals. 
\end{lemma}
\begin{proof}
We say an interval $I$ has rank $r$, if it has performed $(r-1)$ growth steps. Note that its length is then usually $|I|=2^r$, unless of course $I$ is a final interval (then its length is $k+1$).
A process $p$ is contained in at most 2 intervals of rank 1 and at most 4 intervals of rank 2.
Let $I$ and $J$ be two active intervals in $\mathcal{I}$ with rank $r \ge 3$. 
Let $I'$ and $J'$ denote the intervals $I$ and $J$ before their most recent growth step, respectively.  
Assume the distance between $I$'s center and $J$'s center is less than $\frac{|I|}{4}$. Then $I' \subseteq J$ and $J'\subseteq I$. 
But then, during the last growing phase of either $I$ or $J$, one of them should have dominated the other one, a contradiction. 
In case $I$ or $J$ is a monochromatic or dominated interval, we can argue the same argument about the previous growing phases of $I$ and $J$, where both of them were active. 
Then, the centers of $I$ and $J$ are at least at a distance of $\frac{|I|}{8}$ apart.

This means a process $p$ is contained in at most 8 intervals $I$ of rank $r \ge 3$. Since there are at most $\log k + 1$ ranks, the lemma follows.   
\end{proof}

\begin{lemma}\label{lemma:opts_lower_bound}
The cost of the optimal static algorithm $\OPT$ is at least 
$\frac{1}{\mcO(\log k)}\sum_{I}\OPT(I)$.
\end{lemma}
\begin{proof}
For a process $p$ we denote by $\OPT(p)$ the cost that the optimal algorithm
experiences for migrating process $p$. 
Further, for an edge $e$ we use $\OPT(e)$ to denote the cost that the optimal algorithm experiences for communicating along $e$.
Then, $\OPT = \sum_{p} \OPT(p) + \sum_{e} \OPT(e)$.
\def\dsum{\textstyle\sum}
\begin{equation*}\begin{split}
\dsum_{I}\OPT(I) &= \dsum_{I} (\dsum_{p \in I} \OPT(p) + \dsum_{e \in I} \OPT(e))\\
&= \dsum_{p} |\{I: p \in I\}|  \OPT(p) + \dsum_{e} |\{I: e \in I\}|  \OPT(e) \\
&\le \mcO(\log k) \dsum_{p} \OPT(p) + \mcO(\log k)\dsum_{e} \OPT(e)\\
&= \mcO(\log k) \OPT
\end{split}\end{equation*}
The inequality is due to Lemma~\ref{lemma:contained_in_log}
\end{proof}

Now, we summarize the overall cost of our online algorithm.

\thmtwo*
\begin{proof}
Let $\epsilon' := \min\{\frac{\epsilon}{2}, 1\}$ and $\dbar := \max\{\frac{2}{2+\epsilon}, \frac{14}{15}\}$. We execute the slicing procedure with $\dbar$ and the scheduling procedure with $\epsilon'$ parameters.
\def\dsum{\textstyle\sum}
Due to Observation~\ref{obs:hit_move_costs} and Lemma~\ref{lem:optlower} we
derive that $\costhit + \costmove \le
\dsum_I \cost(I) \le \frac{2}{1-\dbar}\dsum_I \OPT(I).$
Furthermore, due to Lemma~\ref{lemma:merge_cost} we 
know that $\costmerge $ is at most $ \mcO((1-\dbar)^{-1}\log k)\sum_{I} \OPT(I)$. 
Lemmas~\ref{lemma:cost_migrations} and~\ref{lemma:col_cost} state that
$\costmono + \costbal \le \mcO({1}/{\epsilon})(\costmove + \costmerge).$
Let $\operatorname{cost}_{\operatorname{onl}}$ denote the overall costs of the online algorithm. 
Then, 
\begin{align*}
\operatorname{cost}_{\operatorname{onl}} &=\costhit + \costmove + \costmerge + \costmono + \costmig\\ 
&\le \mcO((\epsilon(1-\dbar)^{-1}\log k)\dsum_{I \in \mathcal{I}} \OPT(I)\\
&\le \mcO((\epsilon^{-2}\log k)\dsum_{I \in \mathcal{I}} \OPT(I), \qquad \text{ for } \dbar  \ge {2}/{2+\epsilon'}\enspace.
\end{align*}
Applying Lemma~\ref{lemma:opts_lower_bound} ($\sum_{I} \OPT(I) \le \mcO(\log k)\OPT$), we derive that our costs are at most 
$\mcO(\log^2(k)/\epsilon^2)\OPT$.

Lemma~\ref{lemma:correctness_scheduling} states that for $\dbar \ge {2}/(2+\epsilon')$ the scheduling algorithm places at most $3+2\epsilon'$ = $3+\epsilon$ processes on each server, which concludes the proof.
\end{proof}

\section{Conclusion}\label{sec:conclusion}

We presented the first polylogarithmically-competitive 
online algorithms for the balanced graph partitioning problem
for a scenario where the communication pattern inherently and continuously
requires inter-server communications and/or migrations.
In particular, we described two different approaches, one for the static
model and one for the dynamic one, and proved their competitiveness
accordingly.

Our work opens several interesting directions for future research.
In particular, it would be interesting to study lower bounds on the achievable competitive ratio in the two models. The main open question regards whether polylogarithmic-competitive algorithms can also be achieved under more general communication patterns.

\section*{Acknowledgements}\label{sec:ack}
Research supported by German Research Foundation (DFG),  grant 470029389 (FlexNets), 2021-2024
and 
Federal Ministry of Education and Research (BMBF), grant 16KISK020K (6G-RIC), 2021-2025.

\bibliographystyle{unsrt}
\bibliography{bibliography}

\clearpage
\appendix
\section{Smooth Minimum Approximation}\label{appendix}
Let $x=(x_1, x_2, \dots, x_n)^T \in \mbbR^n$ be an $n$ dimensional vector. The function 
	$\smin(x) := -\ln (\textstyle\sum_i e^{-x_i})$ 
smoothly approximates the minimum function $\min(x)=\min(x_1, x_2, \dots, x_n)$.


\begin{fact} \phantom{t}
\label{fact:smin}
\begin{enumerate}[(i)]
	\item The function $\smin(x)$ approximates the minimum up to an additive term: \label{lemma:test}
	\begin{align*}
	\min(x) - \ln n \le \smin(x) \le \min(x)
	\end{align*}
	\item The gradient $\nabla\smin(x)$ is a probability distribution.
\end{enumerate}
\end{fact}
\begin{proof}\phantom{t}
	\begin{enumerate}[(i)]
		\item Let $x_*=\min(x)$. Using facts $\sum_i e^{-x_i} \ge e^{-x_*}$ and $\sum_i e^{-x_i} \le ne^{-x_*} $, we get
			\begin{align*}
				x_* - \ln n = -\ln (n e^{-x_*}) \le -\ln (\textstyle\sum_i e^{-x_i}) \le -\ln (e^{-x_*})=x_* 
			\end{align*}
		\item 
		The $i$-th component of the gradient is $\nabla\smin(x)_i = \frac{e^{-x_i}}{\sum_i e^{-x_i}}\ge 0$. 
		The sum of all components is exactly 1.\qedhere
	\end{enumerate}
\end{proof}

\noindent
The change of the $\smin$ function is well approximated by its gradient. Formally:
\begin{lemma}\phantom{x}
\begin{enumerate}[(i)]\parskip0pt
	\item For all vectors $x,\ell \ge 0$, $\ell_i \le 1$,\label{lemma:smin1}
	\begin{equation*}
		\smin(x+\ell)-\smin(x) \ge \tfrac{1}{2} \nabla\smin(x)^T\ell
	\end{equation*}
	\item For all vectors $x, \ell \ge 0$,\label{lemma:smin2}
	\begin{equation*}
		\norm{\nabla\smin(x+\ell)-\nabla\smin(x)}_1 \le 2\nabla\smin(x)^T\ell
	\end{equation*}
\end{enumerate}
\label{lemma:smin}
\end{lemma}
\begin{proof}\phantom{s}
	\begin{enumerate}[(i)]
		\item Let $A=\sum_{i=1}^n e^{-x_i}$. Then,
		\begin{equation*}
                \begin{split}
			\smin(x+\ell)-\smin(x) 
			&= -\ln (\tfrac{1}{A}\textstyle\sum_i e^{-(x_i + \ell_i)}) \\
			&= -\ln (\tfrac{1}{A}\textstyle\sum_i e^{-x_i} +  e^{-x_i}(e^{-\ell_i}-1)) \\
			&= -\ln (1 + \tfrac{1}{A}\textstyle\sum_i e^{-x_i}(e^{-\ell_i}-1)) \\
			&\ge -\tfrac{1}{A}\textstyle\sum\nolimits_i e^{-x_i}(e^{-\ell_i}-1) \\
			&= \tfrac{1}{A}\textstyle\sum\nolimits_i e^{-x_i}(1- e^{-\ell_i}) \\
			&\ge \tfrac{1}{A}\textstyle\sum\nolimits_i e^{-x_i}\frac{\ell_i}{2} = \tfrac{1}{2}\nabla\smin(x)^T\ell 
                \end{split}
		\end{equation*}

		The first inequality follows by applying the fact $\ln(1+z)\le z$. The second inequality follows because $1-e^{-z}\ge \frac{z}{2}$, for $0\le z \le 1$. 
		
		\item 
		Since both $\nabla\smin(x+\ell)$ and $\nabla\smin(x)$ are probability distributions, the sum of their components is 1. Consider the change of the components going from $x$ to $x+\ell$. The sum of increasing components must equal the sum of decreasing components. Now, let $I$ be the set of indices that are decreasing. Consider such a component $i\in I$. Then,
		\begin{align*}
		\nabla\smin(x)_i-\nabla\smin(x+\ell)_i 
		&= \frac{e^{-x_i}}{\sum_j e^{-x_j}} - \frac{e^{-(x_i+\ell_i)}}{\sum_j e^{-(x_j+\ell_j)}}\\
		&\le \frac{e^{-x_i} - e^{-(x_i+\ell_i)}}{\sum_j e^{-x_j}} \\
		&= \frac{e^{-x_i} (1 - e^{-\ell_i})}{\sum_j e^{-x_j}}\\
        &\le \frac{e^{-x_i}\ell_i}{\sum_j e^{-x_j}} \\
		\end{align*}
		In the first inequality we used that since $l \ge 0$, $\sum_j e^{-(x_j+\ell_j)} \le \sum_j e^{-x_j} $  must hold. In the last inequality we used the fact that $1-e^{-\ell_i}\le \ell_i$. Then, 
		\begin{align*}
			\norm{\nabla\smin(x+\ell)-\nabla\smin(x)}_1 
			&= 2 \sum_{i\in I} \nabla\smin(x)_i-\nabla\smin(x+\ell)_i \\
			&\le 2 \sum_{i\in I} \frac{e^{-x_i}\ell_i}{\sum_j e^{-x_j}} \\
			&\le 2 \sum_{i=1}^n \frac{e^{-x_i}\ell_i}{\sum_j e^{-x_j}}\\ 
            &= 2\nabla\smin(x)^T\ell
		\end{align*}
	\end{enumerate}
\end{proof}

Now we generalize the $\smin$ function. We define a function 
\begin{align*}
\fmin_c(x):=c\cdot\smin(\tfrac{1}{c}x). 
\end{align*}
for a constant $c \ge 1$. The advantage of the $\fmin_c(x)$ function over $\smin(x)$ is that we can control the gradients change using the constant $c$. On the other side the approximation of the minimum becomes worse.
The following properties apply to $\fmin_c(x)$:

\begin{lemma}\phantom{bla}
\begin{enumerate}[(i)]
	\item The function $\fmin_c(x)$ approximates the minimum up to an additive term:
	\begin{align*}
	\min(x) - c\ln n \le \fmin_c(x) \le \min(x)
	\end{align*}
	\item $\nabla\fmin_c(x)$ is a probability distribution.
	Furthermore, $$\nabla\fmin_c(x)=\nabla\smin(\frac{1}{c}x)\enspace.$$
	\item For all vectors $x,\ell \ge 0$, $\ell_i \le 1$,
	\begin{align*}
		\fmin_c(x+\ell)-\fmin_c(x) \ge \tfrac{1}{2} \nabla\fmin_c(x)^T\ell
	\end{align*}
	\item For any vectors $x, \ell \ge 0$,
	\begin{align*}
		\norm{\nabla\fmin_c(x+\ell)-\nabla\fmin_c(x)}_1 \le \tfrac{2}{c}\nabla\fmin_c(x)^T\ell
	\end{align*}
\end{enumerate}
\end{lemma}
\begin{proof}
	Let $x'=\frac{1}{c}x$ and $\ell'=\frac{1}{c}\ell$. Since $c\ge 1$, we have that
        $x', \ell' \ge 0$ and $\ell_i \le 1$. 
	\begin{enumerate}[(i)]
	\item Using Fact \ref{fact:smin} we obtain:
	\begin{align*}
		\fmin_c(x) &= c\smin(\tfrac{1}{c}x) \le n \min(\tfrac{1}{c}x) = \min(x)\\
		\fmin_c(x) &= c\smin(\tfrac{1}{n}x) \ge n (\min(\tfrac{1}{c}x) - \ln n) = \min(x) - c\ln n
	\end{align*}
	\item Using Fact \ref{fact:smin} we obtain:
 	\begin{align*}
		\nabla\fmin_c(x)_i &= \frac{\partial \fmin_c(x)}{\partial x_i}\\ 
                           &= c\frac{\partial \smin(x')}{\partial x_i'} \frac{\partial x_i'}{\partial x_i}\\ 
                           &= \frac{\partial \smin(x')}{\partial x_i'}\\
                           &= \nabla\smin(\frac{1}{c}x)_i
	\end{align*}
	\item Using Lemma \ref*{lemma:smin} (\ref*{lemma:smin1}) we obtain:
	\begin{align*}
		\fmin_c(x+\ell)-\fmin_c(x) &= c(\smin(x'+\ell') - \smin(x')) \\
		&\ge c \frac{1}{2} \nabla\smin(x')^T\ell'\\
        &=\frac{1}{2} \nabla\fmin_c(x)^T\ell
	\end{align*}
	\item Using Lemma \ref*{lemma:smin} (\ref*{lemma:smin2}) we obtain:
	\begin{equation*}\begin{split}
		\|\nabla\fmin_c(x+\ell)&-\nabla\fmin_c(x)\|_1\\ 
		&= \norm{\nabla\smin(x'+\ell') - \nabla\smin(x')}_1\\
		&\le 2\nabla\smin(x')^T\ell'\\
        &= \frac{2}{c} \nabla\fmin_c(x)^T\ell
	\end{split}\end{equation*}
	\end{enumerate}
\end{proof}

\end{document}